%% file: main.tex
\documentclass[acmsmall,screen,nonacm]{acmart}
\usepackage{braket}
\usepackage{amsmath}
\usepackage{mathtools}
\usepackage{mathpartir}
\usepackage{bussproofs}
\usepackage{makecell}
\usepackage{stmaryrd}
\usepackage{tcolorbox} % needed for colored boxes
\tcbuselibrary{skins,breakable} % optional, for nicer boxes
\usepackage{pdflscape}
\usepackage{graphicx}
\usepackage{adjustbox} % <— for automatic fitting
\usepackage{listings}
\usepackage{xcolor}
\usepackage{subcaption} % for subfigures
\usepackage{float}
\usepackage{multirow}
\usepackage{colortbl}    % for \arrayrulecolor
\usepackage{soul}

\definecolor{mygreen}{HTML}{1b5e20}

\lstdefinestyle{pseudocode}{
    basicstyle=\ttfamily\scriptsize,          % monospaced font
    keywordstyle=\color{blue!70!black},  % keywords in blue
    commentstyle=\color{gray!70!black},  % comments in gray
    stringstyle=\color{orange!90!black}, % (if you use strings)
    columns=fullflexible,                % no weird spacing
    keepspaces=true,                     % preserve indentation
    showstringspaces=false,              % don't show spaces as dots
    frame=none,                          % no box around code
    xleftmargin=1em,                     % small left indent
    escapeinside={(*@}{@*)},             % allow LaTeX inside
    aboveskip=0.5em,
    belowskip=0.5em,
    breaklines=true,  
    backgroundcolor=\color{white},
    frame=single,
}

\input{macros}

\newtheorem{example}{Example}
\newtheorem{definition}{Definition}
\newtheorem{lemma}{Lemma}
\newtheorem{proposition}{Proposition}

\newtheorem{theorem}{Theorem}
\newtheorem*{remark}{Remark}

\AtBeginDocument{%
  }

\begin{document}

%%
%% The "title" command has an optional parameter,
%% allowing the author to define a "short title" to be used in page headers.
\title{Formal Verification of Continuous-Variable Quantum Programs}

%%
%% The "author" command and its associated commands are used to define
%% the authors and their affiliations.
%% Of note is the shared affiliation of the first two authors, and the
%% "authornote" and "authornotemark" commands
%% used to denote shared contribution to the research.
\author{Stefanie Muroya}
\affiliation{%
  \institution{Institute of Science and Technology in Austria (ISTA)}
  \city{Klosterneuburg}
  \country{Austria}
}

\author{Thomas A. Henzinger}
\affiliation{%
  \institution{Institute of Science and Technology in Austria (ISTA)}
  \city{Klosterneuburg}
  \country{Austria}
}

%%
%% By default, the full list of authors will be used in the page
%% headers. Often, this list is too long, and will overlap
%% other information printed in the page headers. This command allows
%% the author to define a more concise list
%% of authors' names for this purpose.
\renewcommand{\shortauthors}{Muroya et al.}

%%
%% The abstract is a short summary of the work to be presented in the
%% article.
\begin{abstract}
Continuous-variable quantum computing (CQC) is a promising paradigm already accessible through public cloud services. 
Yet, unlike discrete-variable quantum computing (DQC), CQC lacks a formal program logic. Extending existing quantum Hoare logics to CQC poses fundamental challenges: continuous-variable quantum programs act on {\em infinite-dimensional} Hilbert spaces;
their measurement outcomes are often {\em unbounded}, so expectation values may {\em diverge}. 
We overcome these challenges to give a formal semantics to a universal programming language for CQC and, to our knowledge,
provide the first Hoare logic for CQC.
We introduce an assertion language based on polynomials over canonical observables and establish soundness and relative completeness of the resulting logic.
We also implement a symbolic weakest-precondition calculator for CQC.
Our tool has successfully verified CQC algorithms from textbooks and calculated their approximation errors for physically realizable implementations,
proved the correctness (i.e., equivalence) of gate decompositions for CQC hardware, 
and computed the resource requirements (i.e., number of photon-number states) for achieving a desired accuracy in the classical simulation of continuous-variable quantum programs.
\end{abstract}

%%
%% The code below is generated by the tool at http://dl.acm.org/ccs.cfm.
%% Please copy and paste the code instead of the example below.
%%


%%
%% Keywords. The author(s) should pick words that accurately describe
%% the work being presented. Separate the keywords with commas.

%%
%% This command processes the author and affiliation and title
%% information and builds the first part of the formatted document.
\maketitle

\input{sections/_1_intro}
\input{sections/design}
\input{sections/_3_background}
\input{sections/_5_pl}

\input{sections/_6_hoare}

\input{sections/_7_experiments}
\input{sections/conclusion}
% \section*{Data-Availability Statement}
% Our weakest-precondition calculator (Section~\ref{sec:experiments}) will be made publicly available and submitted for Artifact Evaluation.

%%
%% The acknowledgments section is defined using the "acks" environment
%% (and NOT an unnumbered section). This ensures the proper
%% identification of the section in the article metadata, and the
%% consistent spelling of the heading.
\begin{acks}

\end{acks}
%%
%% The next two lines define the bibliography style to be used, and
%% the bibliography file.
\bibliographystyle{ACM-Reference-Format}
\bibliography{sample.bib}

\newpage
\input{sections/appendix.tex}
\end{document}

%% file: macros.tex
\definecolor{hoareBg}{HTML}{ffffff}     % box background
\definecolor{hoareFrame}{HTML}{C3CCD6}  % frame & separator rules
\definecolor{hoareKw}{HTML}{1A57A8}     % keywords
\definecolor{hoareNum}{HTML}{9AA4B0}    % line numbers
\definecolor{commentColor}{HTML}{546e7a}

\newcommand{\assert}[1]{\left\{\,#1\,\right\}}
\newcommand{\kw}[1]{\textcolor{hoareKw}{\mathtt{#1}}}

\newcommand{\mvals}[1]{\vec{#1}}
\newcommand{\defeq}{\vcentcolon=} 

\newcommand{\Ins}[1]{\kw{#1}}
\newcommand{\cwrite}[3]{#1[#2 \leftarrow #3]}
\newcommand{\tohstate}[2]{\left\langle #1,\;#2\right\rangle}
\newcommand{\hstateAt}[2]{#1[#2]}
\newcommand{\csubs}[3]{#1[#2\mapsto #3]}
\newcommand{\subs}[3]{\csubs{#1}{#2}{#3}}
\newcommand{\atomsubs}[2]{\texttt{AtomS}(#1, #2)}
\newcommand{\Supp}[1]{\texttt{Supp}(#1)}

\newcommand{\kset}[0]{K}

\newcommand{\instructions}[0]{\mathcal{K}}
\newcommand{\instruction}[0]{\kset}
\newcommand{\assertions}[0]{\mathcal{A}}
\newcommand{\qterms}[0]{\mathcal{T}}
\newcommand{\cf}[0]{\gamma}
\newcommand{\Terms}[1]{\texttt{Terms}(#1)}

\newcommand{\Nats}[0]{\mathbb{N}}
\newcommand{\mNats}[1]{\Nats^{#1}}
\newcommand{\Reals}[0]{\mathbb{R}}
\newcommand{\mReals}[1]{\Reals^{#1}}
\newcommand{\C}[0]{\mathbb{C}}
\newcommand{\denseDomain}[0]{\mathcal{S}}
\newcommand{\deq}[0]{=_{\denseDomain}}
\newcommand{\fockBasis}[0]{\mathcal{B}}
\newcommand{\unitaries}[0]{\mathcal{U}}

\newcommand{\myHL}[0]{\textsc{CvQHL}}
\newcommand{\myPL}[0]{\textsc{CvQPL}}

\newcommand{\im}[0]{\textbf{i}}
\newcommand{\Id}[0]{\mathbb{I}}
\newcommand{\zeroOp}[0]{\mathbb{O}}
\newcommand{\numModes}[0]{N}

\newcommand{\hermitePol}[0]{H}
\newcommand{\fockBF}[0]{f}

\newcommand{\real}[0]{r} % real variable
\newcommand{\cnum}[0]{z} % complex variable
\newcommand{\nat}[0]{n} % nonnegative integer
\newcommand{\pos}[0]{x} % position 
\newcommand{\mposv}[0]{\mvals{\pos}}
\newcommand{\mnat}[0]{\mvals{\nat}}
\newcommand{\modeI}[0]{i}
\newcommand{\modeII}[0]{j}

\newcommand{\norm}[1]{||#1 ||}
\newcommand{\innerProd}[2]{\innerProdF{#1}{#2}}
\newcommand{\innerProdF}[2]{\langle #1, #2 \rangle}
\newcommand{\conj}[1]{\overline{#1}}
\newcommand{\qvec}[0]{\alpha}
\newcommand{\hs}[0]{\mathcal{H}}
\newcommand{\basis}[0]{e}
\newcommand{\basisSet}[0]{B}

\newcommand{\domainOp}[1]{D(#1)}

\newcommand{\toproj}[2]{\ket{#1}_{#2}\bra{#1}}
\newcommand{\op}[0]{L}
\newcommand{\dm}[0]{\rho}
\newcommand{\Tr}[0]{\texttt{Tr}}
\newcommand{\expectedV}[1]{\langle #1 \rangle}
\newcommand{\allOps}[0]{\mathcal{L}}
\newcommand{\allOpsHs}[0]{\allOps_{*}}
\newcommand{\allOpsAHs}[0]{\allOpsHs^\dagger}
\newcommand{\allOpsA}[0]{\allOps^{\dagger}}

\newcommand{\alldm}[0]{\mathcal{P}}
\newcommand{\alldmS}[0]{\alldm_{\denseDomain}}

\newcommand{\X}[0]{\widehat{X}}
\newcommand{\POp}[0]{\widehat{P}}
\newcommand{\nop}[0]{\widehat{N}}
\newcommand{\dop}[0]{\widehat{A}}
\newcommand{\cop}[0]{\dop^\dagger}
\newcommand{\quadOps}[0]{\Xi}
\newcommand{\obsop}[0]{\widehat{Z}}
\newcommand{\monomial}[0]{\mu}
\newcommand{\monomials}[0]{\mathcal{A}}
\newcommand{\polys}[0]{\mathcal{F}}
\newcommand{\obspolys}[0]{\polys^\dagger}
\newcommand{\poly}[0]{\widehat{Y}}

\newcommand{\charactF}[0]{\chi}
\newcommand{\weylOp}[0]{W}

\newcommand{\allqv}[0]{\mathcal{Q}}
\newcommand{\allcv}[0]{\mathcal{Y}}
\newcommand{\qvar}[0]{q}
\newcommand{\cvar}[0]{y}
\newcommand{\cstate}[0]{c}

\newcommand{\hstate}[0]{h}

\newcommand{\cstates}[0]{\mathcal{C}}

\newcommand{\hstates}[0]{\mathcal{M}}
\newcommand{\hstatesS}[0]{\hstates_{\denseDomain}}

\newcommand{\hobs}[0]{\theta}
\newcommand{\allhobs}[0]{\Theta}

\newcommand{\allhobsS}[0]{\allhobs_{\denseDomain}}

\newcommand{\rexpr}[0]{E}
\newcommand{\bexpr}[0]{B}
\newcommand{\prog}[0]{P}
\newcommand{\aprog}[0]{\prog^{\bullet}}
\newcommand{\progs}[0]{\Sigma}
\newcommand{\aprogs}[0]{\Sigma^{\bullet}}
\newcommand\sem[1]{\llbracket #1\rrbracket}
\newcommand\semD[1]{\llbracket #1\rrbracket^{*}}
\newcommand{\Meas}[0]{\kw{Meas}}
\newcommand{\MeasI}[2]{#1 = \Meas(#2)}
\newcommand{\ResetI}[0]{\Ins{Reset}}

\newcommand{\DispI}[3]{\Ins{D}(#1, #2, #3)}
\newcommand{\Squeeze}[1]{\Ins{S}(#1)}
\newcommand{\SqueezeI}[2]{\Ins{S}(#1, #2)}

\newcommand{\RotI}[2]{\Ins{R}(#1, #2)}

\newcommand{\BSI}[4]{\Ins{BS}(#1, #2, #3, #4)}
\newcommand{\CXI}[3]{\Ins{CX}(#1, #2, #3)}
\newcommand{\CZI}[3]{\Ins{CZ}(#1, #2, #3)}

\newcommand{\CubicPI}[2]{\Ins{V}(#1, #2)}
\newcommand{\Skip}[0]{\Ins{skip}}
\newcommand{\IfElse}[3]{\Ins{if}~(~#1~)~\left\{#2\right\}~\Ins{else}~\left\{~#3\right\}}
\newcommand{\While}[2]{\Ins{while}~\left(#1\right)~ \left\{~#2~\right\}}

\newcommand{\kerterm}[1]{\Lambda(#1)}
\newcommand{\bindterm}[2]{\sum_{#1} #2}
\newcommand{\vacterm}[2]{\nu #1.\, #2} % fresh-ancilla (vacuum) binder
\newcommand{\guardterm}[2]{[#1]\cdot #2} % guarded quantum term
\newcommand{\tjudg}[3]{\vdash #1 : #2 \Leftarrow #3} % backward term judgement
\newcommand{\ancv}[0]{\overline{\allqv}} % ancilla variables
\newcommand{\anc}[0]{a} % ancilla variable
\newcommand{\polyterm}[0]{Y}
\newcommand{\qterm}[0]{\tau}

\newcommand{\pre}[0]{\phi}
\newcommand{\post}[0]{\psi}
\newcommand{\assertion}[0]{\pre}
\newcommand{\weakp}[2]{wp(#1, #2)}
\newcommand{\hltriple}[3]{\left\{#1\right\}~#2~\left\{#3\right\}}

\newcommand{\hsDim}[0]{N}
\newcommand{\progName}[1]{\textbf{\textsf{#1}}}

%% file: sections/_1_intro.tex
\section{Introduction}
Quantum computing is being pursued through multiple hardware technologies and computational paradigms~\cite{Choe22}.
We can divide the gate-based approaches into two types: 
{\em discrete variable} (DQC) and {\em continuous variable} (CQC) quantum computing.
Both paradigms are candidates for bringing quantum advantage, and CQC has its own distinct features and benefits
\cite{Killoran19,Weedbrook12,Zhong20,Graves14,graves16,Bourassa21,larsen25,aghaee25,Braunstein05}.
Proposals for fault-tolerant~\cite{Gottesman00,Renault25} and mixed discrete-continuous~\cite{lloyd03,Wang01} quantum architectures advocate continuous-variable systems as a promising candidate.
Despite these developments and the fact that 
several emerging hardware platforms rely on CQC \cite{Killoran19,Fluhmann19,Sivak23,Putterman25,Joshi21,CampagneIbarcq20,stavenger22}, 
formal verification research has---so far---focused on discrete-variable quantum programs.
In particular, Hoare logic has become an established framework for reasoning about DQC correctness
\cite{hondt06,ying24,zhou19,liu19,sun24Q,Sundaram22}.
In fact, the first DQC logics~\cite{chadha06, Feng07,kakutani09,Ying12} date back to a time before corresponding hardware became widely accessible~\cite{ibm16}.
In contrast, mature CQC hardware and software tools are already accessible~\cite{Killoran19,stavenger22}, 
yet---to the best of our knowledge---no logic has been developed for CQC.
We present and overcome the technical difficulties that lie behind this gap.

In DQC the basic unit of information is the {\em qubit}, and 
a pure state of an $n$-qubit system is a unit vector residing in a {\em finite-dimensional} complex Hilbert space of dimension~$2^n$~\cite{qc_bible}.
In CQC the basic unit of information is the {\em qumode}, and each pure state of an $n$-qumode system is a unit vector in an {\em infinite-dimensional} complex Hilbert space~$\hs$, whose vectors are realized by square-integrable functions (called {\em wave  functions}) of arity~$n$, which map $n$ reals to a complex amplitude~\cite{Braunstein05}.
In both cases, a (general) quantum state is an ensemble (i.e., probability distribution) of pure quantum states, and such ensembles are represented by {\em density} operators
on the corresponding Hilbert space. 

% {\em Observables.}
As quantum information is connected with the classical world
through measurements, we seek to engineer quantum programs whose measurement results encode a solution to a given problem~\cite{qc_bible}. 
Hence we wish to reason about the expected measurement outcomes of quantum programs,
which are captured by self-adjoint linear operators called {\em observables}~\cite{Ballentine14}.
Most Hoare logics for DQC %\cite{sun24Q,Ying12,ying24,zhou19,liu19,hondt06,Sundaram22}
use observables as assertions:
a Hoare triple $\hltriple{\pre}{\prog}{\post}$ is valid iff for all quantum  states, the expected measurement value $\expectedV{\pre}$ of the precondition $\pre$ in the initial state is at most the expected measurement value $\expectedV{\post}$ of the postcondition $\post$ on the quantum state that results from executing the quantum program~$\prog$.
In a hybrid classical-quantum (cq) setting, DQC logics~\cite{feng21} describe both program states and observables: 
{\em cq-states} as functions that map classical states to density operators, and {\em cq-observables} as functions that map classical states to observables.

The Hoare logics for DQC cannot be directly extended to CQC because they rely on assumptions that do not carry over to infinite-dimensional Hilbert spaces.
In particular, CQC measurements have infinitely many outcomes, and the resulting observables are often unbounded.
Unbounded observables do not exist in finite-dimensional Hilbert spaces, which is the reason why DQC logics use only bounded observables.
While bounded observables guarantee well-defined expectation values, 
unbounded observables may have undefined expectation values, undermining the direct use of expectation-based assertion semantics (see Example~\ref{ex:undefined_exp} in Section~\ref{sec:ideas}).

% The primitive relevant set of DQC observables are the Pauli operators~\cite{qc_bible}.
The primitive CQC observables are the
two unbounded {\em canonical} observables $\X$ and 
$\POp$, for measuring position and momentum, respectively~\cite{Weedbrook12,Ballentine14,Braunstein05}.
Polynomials over the canonical observables play a central role in CQC~\cite{Braunstein05}, analogous to the role of Pauli operators in DQC~\cite{qc_bible}: they provide a natural algebraic language for describing both observables and universal CQC instruction sets~\cite{Weedbrook12,Braunstein05}.
A standard orthonormal basis for DQC is the {\em computational basis}, which for one qubit consists of two vectors denoted by $\ket{0}$ and $\ket{1}$
(using Dirac's braket notation):
its projectors $\toproj{0}{}$ and $\toproj{1}{}$ 
are bounded observables that are useful for describing measurement outcomes.
In CQC, the analogue of the computational basis is the {\em Fock basis}: 
for one qumode, it is a discrete orthonormal set $\{\ket{\nat} \mid \nat \in \Nats\}$ of infinitely many vectors. 
Each Fock-basis (a.k.a.\ {\em photon-number}) state $\ket{\nat}$  has a projector $\toproj{\nat}{}$.
Although these Fock projectors theoretically suffice as observables, they are not 
convenient as assertions for CQC 
(see Example~\ref{ex:squeezed_vacuum} in Section~\ref{sec:ideas}).
Instead, our assertion language represents relevant observables (including the Fock projectors) using finite polynomials over the canonical observables.
{This CQC setting can be made hybrid (cq) similarly to DQC,
mapping classical states to both density operators and observables.}

Like DQC Hoare logics,
%\cite{hondt06,Ying12,ying24,zhou19,liu19,sun24Q,Sundaram22}, 
we leverage two equivalent semantics for Hoare triples.
The two approaches are known as the {\em Schrödinger} picture (which reasons about the forward evolution of a quantum state while fixing observables), and the {\em Heisenberg} picture (which reasons about the backward evolution of observables while fixing the quantum state, using the {\em dual} of a program). 
In particular, the Heisenberg picture evolves observables in the same way that proof rules apply the dual of a program to an observable to obtain the weakest precondition.
While duality does not immediately apply to infinite-dimensional Hilbert spaces, we show that our design decisions enable the computation of weakest preconditions and lead to a relatively complete proof system for CQC.

\subsection{A summary of our design decisions and related work}
Extending Hoare-style reasoning to CQC requires a mathematical foundation that ensures well-defined expectation values and their preservation under program execution.
There are three challenges: 
(I)~giving programs and assertions a well-defined semantics, 
(II)~closure and duality under program execution, and 
(III)~a sufficiently expressive and tractable assertion language for verifying interesting programs.
The three challenges, which are discussed in detail in Section~\ref{sec:ideas},
are overcome by the following five design decisions, which distinguish our approach from DQC logics:
\begin{enumerate}
\item 
  We focus on the quantum states that are induced by the {\em Schwartz subspace}~\cite{schwartz57} of the infinite-dimensional Hilbert space.
  Unlike the entire Hilbert space, 
  the Schwartz space contains only physically realizable quantum states.
  As it is dense, it provides a natural domain for unbounded observables. 
\item 
  The Schwartz space does not guarantee well-defined expectation values in the presence of classical variables and potentially infinite measurement outcomes. 
  We therefore restrict our attention to {\em Schwartz cq-states}, which satisfy rapid-decay conditions, and {\em tempered cq-observables}, whose growth is polynomially controlled. 
  The expected value of any tempered cq-observable on a Schwartz cq-state is absolutely convergent, from which duality follows.
  \item
  We focus on {\em admissible} programs,
  which terminate on Schwartz cq-states and restrict dependence of quantum operations on classical variables so that their Heisenberg evolution exhibits at most polynomial growth.
  Admissibility ensures that program execution preserves Schwartz cq-states and its dual preserves tempered cq-observables.
  Example~\ref{ex:unbounded_squeeze} in Section~\ref{sec:ideas} illustrates the need for admissibility.
  While we provide sufficient syntactic conditions for admissibility, 
  general semantic admissibility may be difficult to check on program loops.
\item 
  We build assertions from {\em finite polynomials} over the canonical CQC observables, using Boolean operators and inequalities.
  This provides a readable syntax for CQC (cf.\ Examples~\ref{ex:admissible_prog} and~\ref{ex:myhl} in Section~\ref{sec:ideas}).
  We show that this assertion language is closed under weakest preconditions.
  %Unlike DQC logics, which compute weakest preconditions semantically using matrix multiplication,
  Polynomial weakest preconditions can be computed 
  automatically by syntactic substitution,
  avoiding the explicit manipulation of infinite-dimensional operators that involve derivatives of transcendental functions (cf.\ Equation~\ref{eq:squeezed_vacuum} in Section~\ref{sec:ideas}).
\item 
  We interpret logical {\em implication as subset relation} on sets of cq-states.
  The subset order can cope with unbounded observables and provides a complete lattice for reasoning about implication,
  in contrast to the antilattice induced by the L\"owner order \cite{hondt06,Kadison51,Ramanantoanina24} used in DQC.
  This observation leads to a sound and relatively complete proof system for CQC.
\end{enumerate}
All five design decisions are prompted by the complexities of CQC. 
Together, they provide a semantic foundation for sound and automatable weakest-precondition reasoning for CQC.

\subsubsection*{Related work}
The two closest lines of work focus on reasoning about the equivalence of quantum processes.
The first is a relational Hoare logic with infinite-dimensional operators that are bounded from below~\cite{barthe26}.
While this work provides a foundational theory for relational reasoning in infinite dimensions, it does not specifically address CQC primitives.
In contrast, our work focuses on weakest-precondition reasoning for an explicit, {\em universal} programming language for CQC whose specifications are built from the canonical observables of CQC.
The second related line of work are the
studies on the ZX calculus for CQC~\cite{Shaikh24,Nagayoshi25}.
The basic judgment of this calculus is the equality of linear maps 
(a.k.a.\ quantum circuits).
The ZX calculus is complete for the Gaussian fragment of CQC,
which is a tractable fragment of CQC analogous to Clifford gates for DQC.
However, unlike our programming language,
the Gaussian fragment is not universal for CQC.
Moreover, we reason about general properties of quantum states rather than 
specifically the equivalence of quantum circuits, 
and our polynomial assertions generally lie outside the scope of both related works.

\subsection{Contributions} 
We develop a formal framework for CQC, which consists of 
(A)~a universal continuous-variable programming language that includes control flow, photon-number measurements, and resets to the vacuum, and
(B)~an expressive assertion language and corresponding Hoare logic,
which guarantees finite expectation values for all assertions on all reachable program states. 
(C)~We also implemented this framework in a tool for computing weakest preconditions and present several use cases.

\smallskip\noindent (A)
Our first contribution is a {\em universal\footnote{We consider the criteria for universality given in~\cite{Lloyd99}.}continuous-variable programming language} for CQC with a formal semantics (Section~\ref{sec:pl}).
The instructions of our language $\myPL$ are chosen so that, first,
the forward semantics of admissible programs stays within the set of Schwartz cq-states and second, tempered cq-observables are closed under the dual of every admissible program.
As far as we know, existing CQC programming languages~\cite{Killoran19,Rueda21,stavenger22} lack formal semantics.

\smallskip\noindent (B)
Our second contribution is the {\em first quantum Hoare logic for CQC} (Section~\ref{sec:hoare}).
Our assertion language based on polynomials 
supports reasoning about relevant bounded and unbounded observables. 
Our semantic framework guarantees well-defined expectation values for the supported observables and admissible programs.
We first show how the algebraic structure of the canonical CQC observables allow us to compute weakest preconditions by syntactic substitution.
We then present a proof system and establish its soundness and relative completeness.

\smallskip\noindent (C)
Our third contribution is a 
{\em symbolic calculator of weakest preconditions for continuous-variable quantum programs}, 
and its use in three case studies (Section~\ref{sec:experiments}).
We automated our framework in Python with SymPy, 
which derives a symbolic weakest precondition for a given $\myPL$~program and postcondition. 
Our experimental evaluation of the tool performs three case studies: 
\begin{enumerate}
    \item \textbf{Verification of textbook CQC algorithms.} We verify four published algorithms:
    (i)~homodyne measurement via photon counting \cite{Leonhardt95}, 
    (ii)~a continuous-variable Deutsch-Jozsa oracle~\cite{Jozsa92, Braunstein03}, (iii)~CQC superdense coding~\cite{Ban99,Braunstein00}, and 
    (iv)~CQC state teleportation~\cite{Furusawa11,Braunstein98}. 
    Beyond functional correctness, our tool derives symbolic expressions for the mean and variance of the relevant observables, yielding closed-form characterizations of the intrinsic approximation error present in physically realizable implementations. 
    These derivations confirm results from the physics literature~\cite{Braunstein05} through formal Hoare-style reasoning.
    \item \textbf{Equivalence of continuous-variable quantum programs.} 
    By the Stone–von Neumann theorem and Schur's lemma, two unitary CQC operators induce the same evolution on all canonical observables if and only if they differ by a global phase~\cite{Neumann31,Ballentine14,Hall13}.
    More specifically, if two unitary programs agree on the weakest preconditions for the two postconditions $\X=x$ and $\POp=p$, for symbolic real-valued constants $x$ and~$p$, then the two programs agree on all Hoare triples from $\myHL$.
    In our experiments, we use equivalence checking to verify all quantum-gate decompositions provided in the Strawberry Fields documentation~\cite{Killoran19}.
    \item \textbf{Resource estimation for classical simulation.} 
    The symbolic mean and variance of the number operator can be obtained by computing weakest preconditions;
    they bound through a one-sided Chebyshev-Cantelli inequality~\cite{Boucheron13} the number of
    photon-number states one must retain to simulate a quantum program within a specified target accuracy.
    Our tool can automatically derive rigorous Fock-space truncation bounds for continuous-variable quantum programs using weakest-precondition reasoning. Unlike approaches that analyze truncation errors for specific physical models~\cite{Veitchetal13,Hatomura22,Ohetal21,Qietal20}, our method derives symbolic bounds that 
    allow the truncation of the infinite-dimensional Hilbert space to a finite-dimensional space on which DQC can act, and it enables the entire DQC formal-methods toolchain to be reused for CQC with soundness guarantees.
\end{enumerate}

%% file: sections/design.tex
\section{Motivating examples behind designing a logical foundation for CQC}
\label{sec:ideas}
To establish logical foundations for CQC, 
three challenges (I)--(III) must be overcome:

\subsubsection*{\textnormal{(I)} Giving programs and assertions a well-defined semantics.} 
Qumodes live in an infinite-dimensional Hilbert space 
and continuous-variable operators act on this space.
In particular, the expected value of an observable on a pure quantum state is defined as an improper integral (or infinite series),
which may not converge~\cite{Ballentine14}.
Although mathematically valid, pure quantum states with divergent or undefined expected values are not realizable physically~\cite{Carcassi25}.
By contrast, 
in DQC the finiteness of all sums guarantees well-defined expected values for all assertions on all quantum states~\cite{Reed12}.

\begin{example}[Undefined expectation value] 
Consider a system that consists of a single qumode
whose state is $\qvec(x) = \frac{1}{\sqrt{2}\cdot(1+|x|)}$ in the position basis; 
i.e., the probability density that the qumode is in position $x \in \Reals$ is $\qvec(x)\cdot\qvec(x)$.
Recall that $\X$ is the observable to measure the qumode's position.
If the assertion $\X$ is interpreted over all quantum states in the 
infinite-dimensional Hilbert space~$\hs$, it is ill-formed,
because a position measurement does not have a defined expected value:
$$
\int_{x\in \Reals} x\cdot\qvec(x)\cdot\qvec(x) dx = \int_{x\in\Reals}\frac{x}{2\cdot (1+|x|)^2}\,dx = \infty - \infty.
$$
\label{ex:undefined_exp}
\end{example}
{{The first challenge}} for building logical foundations of CQC, therefore, is to identify
\begin{enumerate} 
\item 
a subset $\mathcal{S}\subset\hs$ of pure quantum states that are physically realizable, 
\item 
a relevant set $\hstatesS$ of cq-states with possibly infinite support, and
\item a relevant set $\allhobsS$ of cq-observables with possibly infinite support,
\end{enumerate}
such that the expected values of all cq-observables in $\allhobsS$ on all cq-states in $\hstatesS$ are well-defined.
{{Our first design decision}} is to use
for $\mathcal{S}$ the {\textit{Schwartz space}}~\cite{Becnel15}---a dense subspace of the infinite-dimensional Hilbert space.
Identifying the Schwartz space alone does not automatically induce a well-behaved set of cq-states and cq-observables, because measurements have infinitely many outcomes.
{{Our second design decision}} is therefore to use rapidly decaying {\textit{\em Schwartz cq-states}} for $\hstatesS$, and {\textit{\em tempered cq-observables}} for $\allhobsS$.
The rapid-decay condition on Schwartz cq-states ensures that the expectation values of polynomial observables decay faster in the classical state than any inverse polynomial. 
Tempered cq-observables have expectation values on Schwartz vectors that grow at most polynomially with the classical state;
hence the expected value of a tempered cq-observable on a Schwartz cq-state is always finite.

\subsubsection*{\textnormal{(II)}~Closure and duality under program execution.}
Ensuring that expectation values are well-defined before program execution is insufficient; program execution must also preserve this property. In particular, CQC instructions whose parameters depend on measurement outcomes may amplify increasingly unlikely outcomes and lead to divergent expectation values.

\begin{example}[Inadmissible classically-parametrized quantum instruction]
\label{ex:unbounded_squeeze}
Let $\nop = \sum_{\nat = 0}^{\infty}\nat\cdot\toproj{\nat}{}$ be the observable that measures the expected photon number on a qumode.
Consider a quantum variable~$\qvar$, a classical variable $\cvar$,  and the program $ \prog = (\MeasI{\cvar}{\qvar};\ \SqueezeI{\qvar}{\cvar}),$  
which measures the photon number of $\qvar$, stores the outcome in $\cvar$, and then squeezes $\qvar$ by the measured value.
Let the initial cq-state be $\hstate = \tohstate{0}{\sum_{\nat = 0}^{\infty}2^{-(\nat+1)}\cdot\toproj{\nat}{}}$, and let $\hobs = \bigoplus_{\nat=0}^{\infty}\tohstate{\nat}{\nop}$ be the cq-observable that maps every classical state to the number operator $\nop$.
If $\hstate'$ denotes the cq-state that results from executing $\prog$, then $\hobs$ has an expected value on $\hstate'$ of at least
\[
\sum_{\nat=0}^{\infty}\Tr((\hobs \otimes \hstate')(\nat)) \geq \sum_{\nat=0}^{\infty} 2^{-(\nat+1)}\cdot \sinh^2(\nat) = \sum_{\nat=0}^{\infty} 2^{-(\nat+1)}\cdot \left(e^{2\nat} + e^{-2\nat} - 2\right)/4 =  \infty.
\]
The parameter of the squeeze instruction stores the outcome of an unbounded observable that has infinitely many outcomes.
So even if the expectation value is finite at each classical state, the sum of infinitely many expectation values can diverge. 
The reason is that outcome-dependent squeezing amplifies the expected photon number faster than the outcome probabilities decay.
\end{example}

Therefore, {{the second challenge}} is to identify a set of {\em admissible} CQC programs that support control flow and a universal gate set such that Schwartz cq-states are preserved under forward semantics, and tempered cq-observables are preserved under the (dual) backward semantics.
{{Our third design decision}} 
is to enforce {{\em polynomially bounded expressions}} within quantum instructions.
Polynomially bounded expressions have an absolute value that grows at most polynomially with the classical state.
Forward and backward semantics preserve the corresponding decay conditions, as polynomially bounded functions are closed under sums, products, and composition.

\begin{example}[Admissible classically-parametrized quantum instruction]
\label{ex:admissible_prog}
We use the setup of Example~\ref{ex:unbounded_squeeze}, 
but consider the program
$\prog' = (\MeasI{\cvar}{\qvar};\ \DispI{\qvar}{\cvar}{0})$, which displaces the position of the qumode at $\qvar$ by $\sqrt{2}\cdot\cvar$ and leaves its momentum unchanged.
This is the feed-forward correction frequently used in continuous-variable protocols like quantum  teleportation.
Unlike $\prog$, the program $\prog'$ is admissible.
The cq-observable $\hobs$ has after executing $\prog'$ the expected value 
$
\sum_{\nat=0}^{\infty}\Tr((\hobs \otimes \hstate')(\nat)) 
= \sum_{\nat = 0}^{\infty} (\nat + \nat^2)\cdot 2^{-(\nat+1)} = 4.
$
Since the displacement depends linearly on the measurement outcome, its effect on polynomial observables is polynomially bounded.
In particular, the position operator evolves polynomially under the dual of the displacement instruction in $P'$ 
as $ \X \xrightarrow{\DispI{\qvar}{\cvar}{0}} (\X + \sqrt{2}\cdot\cvar\cdot\Id)$,
whereas the dual of the squeeze instruction in $P$
evolves the position operator exponentially as $\X \xrightarrow{\SqueezeI{\qvar}{\cvar}} (e^{-\cvar}\cdot \X)$.
\end{example}

Loop guards may depend on the classical state, but admissibility requires a uniform bound on the number of loop iterations. This restriction is deliberate: almost-sure termination does not guarantee finite expectation values of unbounded observables. 
The uniform bound is a conservative, sufficient condition that still admits useful program loops whose guards depend on measurement outcomes.

\subsubsection*{\textnormal{(III)}~A sufficiently expressive assertion language with computable weakest preconditions.}
Observables of CQC are generally unbounded,
because measurement outcomes can be arbitrarily large;
e.g., the position measurement $\X$ can be any real number.
Moreover, 
even for a simple bounded postcondition, the weakest precondition is not guaranteed to have an appealing finitary syntax.

\begin{example}[Weakest precondition of vacuum under squeeze] 
Consider the {\em squeeze} instruction $\Squeeze{r}$ 
(with a real parameter~$r$), 
which is useful for generating entangled states in CQC, 
and the bounded postcondition $\ket{0}\bra{0}$,
which measures the probability of the vacuum state $\ket{0}$.
The resulting weakest precondition is the observable 
$\ket{\qvec}\bra{\qvec}$ with 
\begin{equation}
    \label{eq:squeezed_vacuum}
    \ket{\qvec} = 
    \frac{1}{
        \sqrt{\cosh(r)}
    } \sum^{\infty}_{i = 0} (\tanh(r))^i 
        \frac{\sqrt{(2i)!}}{2^i(i!)}\int_{x \in \Reals}
        \frac{ 
            e^{-x^2/2} \cdot \hermitePol_{2i}(x)
        }{
            \pi^{1/4}\sqrt{4^i\cdot (2i)!}
        }\ket{x}dx,
\end{equation}
where $\hermitePol_{i}$ is the $i$-th physicist's Hermite polynomial~\cite{Sakurai20}, and $\ket{x}$ is an eigenvector of the position operator~$\X$.
(There is no need to understand this equation!---it is shown only to illustrate how complicated CQC assertions can become if no attention is paid to the design of the assertion language.)
\label{ex:squeezed_vacuum}
\end{example}
{{The third challenge}} for a CQC Hoare logic, therefore, is to devise a sufficiently rich assertion language, 
which can specify relevant CQC observables that may be  either bounded (like the Fock-state projectors) or unbounded (like the canonical observables),
 such that weakest preconditions are expressible with respect to all admissible programs.
{{Our fourth design decision}}
is to base the assertion language on {\textit{finite polynomials}} over the canonical observables.
Finite polynomials can express Fock state projectors (e.g., $\toproj{0}{} = \lim_{\beta \rightarrow \infty} e^{-\beta\cdot\nop}$, where $\nop=(\X^2 + \POp^2 -\Id)/2$ can be expressed as a finite polynomial).
Moreover, one can characterize a (general) quantum state completely using finite polynomials~\cite{Braunstein05,Weedbrook12,Kok10}, 
and many quantum states of interest can be characterized in terms of the expected values of finitely many monomials~\cite{Weedbrook12}.

\begin{example}[Polynomial weakest precondition of vacuum under squeeze]
Recall Example~\ref{ex:squeezed_vacuum}.
We write the polynomial precondition that specifies the cq-state whose sum of function values is $\ket{\qvec}\bra{\qvec}$, and the polynomial postcondition that defines $\ket{0}\bra{0}$.
Both $\ket{\qvec}$ and the vacuum state $\ket{0}$ are Gaussian states that can be characterized by a covariance matrix over the expected values of finitely many polynomials~\cite{Weedbrook12}:
\begin{equation*}
\hltriple{\X = 0 \land e^{-2r}\cdot\X^2 = 0.5 \land \POp = 0 \land e^{2r}\cdot\POp^2 = 0.5}{\Squeeze{r}}{\X = 0 \land \X^2 = 0.5 \land \POp = 0 \land \POp^2 = 0.5}.
\end{equation*}
These assertions are more readable than Equation~\ref{eq:squeezed_vacuum}.
\label{ex:myhl}
\end{example}

\begin{example}[Triple of an admissible program]
For the program $\prog'$ of Example~\ref{ex:admissible_prog},
our logic can prove the validity of the Hoare triple 
\begin{equation}
    \hltriple{\bindterm{\cvar}{\left(\toproj{\cvar}{}\cdot(\nop + \cvar^2)\right)} = 4}{\MeasI{\cvar}{\qvar};\ \DispI{\qvar}{\cvar}{0}}{\nop = 4}.
    \label{eq:ex_admissible}
\end{equation}
The precondition denotes an infinite sum whose terms are obtained by replacing the classical variable $\cvar$ in the argument by the nonnegative integers. 
The term $\nop + \cvar^2$ in the precondition shows that the expected photon number increases quadratically in the measured outcome when executing the program~$\prog'$.
Our tool computes the precondition\footnote{The Fock-state projector $\toproj{\cvar}{}$ is not a primitive in our assertion language, but represents the polynomial term $\Lambda(1\cdot (\nop - \cvar\cdot\Id)^2)$, where $\Lambda(\cdot)$ is a language primitive denoting a limit of an exponential on the polynomial argument.} automatically from the program and postcondition.
\end{example}

The DQC logics exploit boundedness.
For bounded assertions $\pre$ and $\post$, 
logical implication $\pre\Rightarrow\post$ can be defined as follows: 
on all quantum states, 
the expected measurement value of $\pre$ is at most the expected 
measurement value of~$\post$;
this is called the L\"owner order on observables \cite{hondt06}
and corresponds to the DQC Hoare triple $\hltriple{\pre}{\Skip}{\post}$.
However, 
the L\"owner order does not capture the desired CQC semantics of 
$\hltriple{\pre}{\Skip}{\post}$.
To see this, let $\pre=\X$ and $\post=2\X$.
Then $\hltriple{\pre}{\Skip}{\post}$ shall be valid in CQC, 
but for any quantum state whose expected position is negative,
the expected value of $\X$ is greater than the expected value of $2\X$.
The problem is that the measurement outcome of $\X$ is indefinite
(i.e., positive, zero, or negative).
By contrast,
in DQC, the boundedness of all observables ensures that their measurement outcomes can be made positive semidefinite (i.e., positive or zero),
by adding a lower bound to all observation values. 

{{Our fifth design decision}} is to replace the DQC interpretation of assertions as functions on quantum states, which map each state to an expected measurement value.
Instead, we use a {\textit{set-based interpretation}} for CQC assertions, 
interpreting every assertion as denoting a set of Schwartz  cq-states,
namely, a subset of the set $\hstatesS$ of Schwartz cq-states.
Syntactically, our assertions resemble those of classical Hoare logic over the reals, while referring directly to the physical quantities relevant to CQC.
Semantically, for example, the assertion  
$\guardterm{\cvar=0}{\X} = 5 \land \guardterm{\cvar=0}{\POp} = 3$ 
denotes the set of Schwartz cq-states for which the expected values of measurement outcomes for position and momentum are 5 and~3, respectively, when considering only the Schwartz density operators that are function values of a classical state satisfying the guard $\cvar=0$. 
Similarly,
the polynomial assertion $\X^2=3\wedge\X=1$ is equivalent to $\guardterm{\top}{\X^2}=3\wedge \guardterm{\top}{\X}=1$, and is satisfied by 
the Schwartz cq-states for which the variance 
of the position measurement is $3-1=2$.
This set-based semantics allows us to model logical implication extensionally, as the subset relation on sets of quantum states, 
thus avoiding the L\"owner order.
In particular,
a Hoare triple $\hltriple{\pre}{\prog}{\post}$ is valid if the admissible CQC program $\prog$ transforms every Schwartz cq-state that satisfies the precondition $\pre$, into a cq-state that satisfies the postcondition~$\post$.

%% file: sections/_3_background.tex
\section{Framework}
\label{sec:background}
Unlike in finite dimensions, many CQC operators are unbounded and defined only on dense domains. This section introduces the operator-theoretic notions relevant to our work.
In the following paragraphs, we set Planck's constant, mass and frequency to 1.

\paragraph{Notation.} We denote by $\Nats, \Reals, \C$ the sets of all nonnegative integers, reals, and complex numbers respectively. 
The imaginary unit is denoted by $\im = \sqrt{-1}$ (in bold), and 
the conjugate of a complex number $\cnum \in \C$ by $\conj{\cnum}$.
For a tuple $\mvals{x}$, we denoted by $|\mvals{x}| = \sum_{i = 0}^{\nat-1} \mvals{x}_i$ the sum of its components. 

\subsection{Primer on linear algebra for continuous-variable quantum computing}
\paragraph{Infinite-dimensional Hilbert space $(\hs)$.} A Hilbert space $\hs$ is a complete inner product space~\cite{Ballentine14}.
In continuous-variable quantum computing, an $\numModes$-qumode system resides in a Hilbert space $\hs \simeq  L^2(\mReals{\numModes})$ realized by the set $L^2(\mReals{\numModes}) = \{\qvec: \mReals{\numModes} \rightarrow \C \mid \int_{\mposv \in \mReals{\numModes}}\qvec(\mposv)\cdot \conj{\qvec(\mposv)} < \infty \}$ of square-integrable functions with arity $\numModes$.
The inner-product is defined by $\innerProd{\qvec}{\qvec'} = \int_{\mposv \in \mReals{\numModes}}\conj{\qvec(\mposv)}\cdot \qvec'(\mposv)$ for all vectors $\qvec,\qvec'$ in $\hs$.
The norm of a vector $\qvec$ is formally given by $\norm{\qvec} = \sqrt{\innerProd{\qvec}{\qvec}}$.

\paragraph{The Schwartz space $(\denseDomain)$.}
A {\em dense domain} is a linear subspace of $\hs$ whose closure is $\hs$.
The {\em Schwartz space} is formally defined by
    $\denseDomain = \{\qvec \in C^{\infty}(\mReals{\numModes}) \mid \forall \mvals{m},\mvals{k} \in \mNats{\numModes}.~\sup_{\mposv \in \mReals{\numModes}}|\mposv^{\mvals{m}} (\partial^{\mvals{k}} \qvec)(\mposv)| < \infty\},$
where $C^{\infty}(\mReals{\numModes})$ is the set of infinitely differentiable functions $\mReals{\numModes} \rightarrow \C$, $\mposv^{\mvals{m}} = x_0^{m_0}x_1^{m_1}\cdots x_{\numModes-1}^{m_{\numModes-1}}$, and $\partial^{\mvals{k}} = \partial_{x_0}^{k_0}\partial_{x_1}^{k_1}\cdots \partial_{x_{\numModes-1}}^{k_{\numModes-1}}$ are the partial derivatives~\cite{schwartz57,Becnel15}. Every such function is square-integrable, so $\denseDomain \subset \hs$.
Intuitively, $\denseDomain$ consists of smooth functions whose derivatives of all orders are rapidly decreasing, and it guarantees finite expectations for all polynomials over canonical observables while being dense in $\hs$~\cite{Carcassi25,Becnel15}.

\begin{remark}
In the remainder of this work, we fix the number of qumodes $\numModes$, our Hilbert space $\hs\simeq L^2(\mReals{\numModes})$, the dense domain $\denseDomain$, unless specified otherwise.
\end{remark}

\paragraph{The Fock basis $(\fockBasis)$.}
Let $\hermitePol_{\nat}$ be the $\nat$-th (physicist's) Hermite polynomial~\cite{Sakurai20}. The {\em Fock basis} is a countable set $\fockBasis\subset \denseDomain$ of orthonormal basis vectors for $\hs$~\cite{Becnel15}. Let $\ket{\nat}$ denote the square-integrable function  $\fockBF_\nat(\pos) = (e^{-\pos^2/2} \cdot \hermitePol_{\nat}(\pos))/(\pi^{1/4}\sqrt{2^\nat\cdot \nat!})$~\cite{Ballentine14, Sakurai20}, for all nonnegative integers $\nat \in \Nats$.
Formally, the {\em Fock basis} is given by the set $\fockBasis =  \{\ket{\mnat}=\ket{\mnat_0}\otimes\cdots\otimes\ket{\mnat_{\numModes-1 }}:\mnat\in\Nats^{\numModes} \}$. The Fock basis state $\ket{0}^{\otimes \numModes}$ is often referred to as {\em vacuum}.

\paragraph{Linear operators.}
An operator $\op$ on $\hs$ is a linear map $\op : \domainOp{\op} \rightarrow \hs$ defined on a subspace $\domainOp{\op} \subseteq \hs$, called its domain~\cite{Ballentine14,Reed12,Blanchard15}. For operators given by a formal expression, such as $\X$ or $\POp$, we take the maximal domain, i.e., the set of all vectors $\ket{\qvec}\in\hs$ for which $\op\ket{\qvec} \in \hs$. 
We denote by $\allOpsHs$ the set of all linear operators for $\hs$.
We say that $\op$ is {\em densely defined} if $\domainOp{\op}$ is dense in $\hs$; in particular, this holds if $\denseDomain \subseteq \domainOp{\op}$.
The {\em adjoint} of a densely defined operator $\op$ is the unique operator $\op^\dagger$ whose domain $\domainOp{\op^\dagger}$ is the set of all vectors $\ket{\qvec} \in \hs$ for which there exists a vector $\ket{\qvec'} \in \hs$ such that
$
\innerProdF{\op\ket{\qvec''}}{\qvec} = \innerProdF{\qvec''}{\qvec'}
$
for all $\ket{\qvec''}\in \domainOp{\op}$. We then define the adjoint by its action $\op^\dagger \ket{\qvec} = \ket{\qvec'}$ on all $\ket{\qvec} \in \domainOp{\op^\dagger}$~\cite{Reed12}.
We have $\op$ is {\em bounded} if there exists a constant $C \geq 0$ such that for all vectors $\qvec \in \domainOp{\op}$, holds $\norm{\op\cdot\ket{\qvec}} \leq C\cdot \norm{\qvec}$. 
We denote by $\op\geq \zeroOp$ if $\op$ is {\em positive semidefinite}: for all $\qvec \in \domainOp{\op}$ we have $\bra{\qvec}\cdot\op\cdot\ket{\qvec} \geq 0$.
Given two operators $\op, \op'$ with $\denseDomain \subseteq \domainOp{\op}\cap\domainOp{\op'}$, we define {\em weak equality} on $\denseDomain$ as $\op \deq \op'$ if $\op\ket{\qvec} = \op'\ket{\qvec}$ for all vectors $\ket{\qvec} \in \denseDomain$.
In contrast, {\em strong equality} depends on the entire Hilbert space: we write $\op = \op'$ if $\domainOp{\op} = \domainOp{\op'}$ and $\op\ket{\qvec} = \op'\ket{\qvec}$ for all $\ket{\qvec} \in \domainOp{\op}$. 
If $\op$ and $\op'$ are bounded, then $\op \deq \op'$ already implies $\op = \op'$, since $\denseDomain$ is dense. 
The {\em commutator} of the two linear operators is $[\op, \op'] = \op\op' - \op'\op$.

In this work, we consider the set $\allOps = \{\op : \domainOp{\op} \rightarrow \hs \mid  \forall \ket{\qvec}\in \denseDomain.~\op\ket{\qvec}, \op^\dagger\ket{\qvec} \in \denseDomain\} \subset \allOpsHs$ of densely defined operators, closed under the adjoint, and where the Schwartz space is invariant. 
Restricted to $\denseDomain$, the set $\allOps$ is closed under sums, products and adjoints~\cite{Schmudgen90}.

\paragraph{Observables.}
An {\em observable} is a linear operator $\obsop \in \allOps$ acting on $\hs$ that is {\em self-adjoint}, i.e., $\obsop = \obsop^\dagger$. 
Self-adjoint operators guarantee that the expectation value is a real number. We denote the set of all self-adjoint operators on $\hs$ by $\allOpsAHs$, and by $\allOpsA = \{\obsop \in \allOps \mid \obsop = \obsop^\dagger\} \subset \allOpsAHs$ those observables that leave the Schwartz space invariant.

\paragraph{Traces.}
Let $\op$ be a linear operator for $\hs$ and $\basisSet = \{\ket{\basis_i}\}_i$ be an orthonormal basis for $\hs$, the {\em trace} of $\op$ with respect to $\basisSet = \{\ket{\basis_i}\}_i$ is denoted by $\Tr_{\basisSet}(\op) = \sum_{i=0}^{\infty}\bra{\basis_i} \op \ket{\basis_i}$.
The {\em trace-norm} of $\op$ is given by $\Tr_{\basisSet}(|\op|) = \sum_{i=0}^{\infty}\bra{\basis_i} \cdot \sqrt{\op^\dagger\cdot \op} \cdot \ket{\basis_i} < \infty$; and 
 $\op$ is {\em trace-class} if $\Tr_{\basisSet}(|\op|) < \infty$ for all orthonormal bases $\basisSet$ for $\hs$~\cite{Reed12}. Traces of trace-class operators always converge to the same value regardless of the chosen basis. Therefore, for a trace-class operator $\op$ we simply write $\Tr(\op)$.
 Given a sequence of trace-class operators $\{\op_n\}_{n \in \Nats}$, we say that $\op_n$ {\em converges in trace-norm} to a trace-class operator $\op$ if $\lim_{n \rightarrow \infty}\Tr(|\op_n - \op|) =0$.

\paragraph{Canonical operators for qumodes.}
A qumode is a {one-dimensional quantum harmonic oscillator (QHO)}~\cite{Braunstein05,Killoran19,Weedbrook12}, defined by the two canonical observables $\X$ and $\POp$---collectively referred to as {\em quadratures}, and individually as {\em position} and {\em momentum} respectively---that satisfy the {canonical commutation relation} $[\X, \POp] = \im\cdot\Id$. For all vectors $\qvec \in \domainOp{\X}$, the position operator $\X$ is defined as $(\X(\qvec))(x) = x\cdot \qvec(x)$. The standard representation of the momentum operator consistent with the commutation relation is $\POp = -\im\frac{d}{dx}$.
Alternatively, a QHO can also be defined via the two canonical operators:
$\dop = \sum_{i=0}^{\infty}\sqrt{i+1} \ket{i}\bra{i+1}$ and $
\cop = \sum_{i=0}^{\infty}\sqrt{i+1}\ket{i+1}\bra{i},$
called {\em ladder operators}, that satisfy {\em the canonical commutation relation} $[\dop, \cop] = \Id$. 
Since $\dop$ maps $\ket{i+1}$ to $\sqrt{i+1}\ket{i}$, and $\cop$ maps $\ket{i}$ to $\sqrt{i+1}\ket{i+1}$, they are also referred to as {\em annihilation} and {\em creation} operators respectively. The number operator  $\nop = \cop\cdot\dop = \sum_{i \in \Nats} i \ket{i}\bra{i} = (\X^2 + \POp^2 - \Id)/2$ is the observable that we use in assertion to describe the average photon number. The number operator can be expressed as a polynomial over ladder operators using the identities from Equation~\ref{eq:id_canon_ops}.
\begin{equation}
    \X = \frac{1}{\sqrt{2}}(\dop + \cop),\quad\quad 
    \POp = \frac{1}{\im \sqrt{2}}(\dop - \cop),\quad\quad
    \dop = \frac{1}{\sqrt{2}}(\X + \im\cdot\POp), \quad\quad 
    \cop = \frac{1}{\sqrt{2}}(\X - \im\cdot\POp).
    \label{eq:id_canon_ops}
\end{equation}

A multi-qumode system of $\numModes$ qumodes is realized by $\numModes$ one-dimensional quantum harmonic oscillators~\cite{Braunstein05,Killoran19,Weedbrook12}. The $\modeI$-th qumode has canonical observables $\X_\modeI = \Id^{\otimes \modeI}\otimes \X\otimes \Id^{\otimes \numModes - \modeI -1}$ and 
$\POp_\modeI = \Id^{\otimes \modeI}\otimes \POp\otimes \Id^{\otimes \numModes - \modeI -1}$. These canonical operators satisfy the following {\em canonical commutation relations}: (i) $[\X_\modeI, \POp_\modeI] = \im\cdot\Id$, and (ii) $[\X_\modeI, \X_\modeII] = [\POp_\modeI, \POp_\modeII] =[\POp_\modeI, \X_\modeII] = 0$ for all $0 \leq \modeI,\modeII < \numModes$ and $\modeI\neq \modeII$. We denote by $\quadOps = \{\X_\modeI, \POp_\modeI \mid 0 \leq \modeI < \numModes\}$ the set of all quadrature operators.
The Fock-state projector of the $\modeI$-th qumode is denoted by  $\toproj{\nat}{\modeI} = \Id^{\otimes \modeI}\otimes \toproj{\nat}{}\otimes \Id^{\otimes \numModes - \modeI -1}$ for all $\nat \in \Nats$. The {\em total number operator} is given by $\sum_{\modeI = 0}^{\numModes-1}\nop_{\modeI}$.

\paragraph{Monomials and polynomials over canonical operators.} 
A {\em monomial} is a linear operator that can be written as a product of the canonical operators of the $\numModes$ qumodes.
We define monomials using ladder operators in a normal order, since we can always use the identities from Equation~\ref{eq:id_canon_ops}  in addition to the commutation relations to rearrange factors.
Let $\mvals{d}, \mvals{c} \in \mNats{\numModes}$ be tuples of $\numModes$ non-negative integers.
A {\em monomial} over canonical operators is a linear operator of the form $\monomial_{\mvals{d}, \mvals{c}} = \prod_{0 \leq \modeI  < \numModes} (\dop_\modeI)^{\mvals{d}_\modeI}(\cop_\modeI)^{\mvals{c}_\modeI}$.
A {\em polynomial $\poly$ over canonical operators} is a finite complex linear combination over monomials. We denote by $\monomials$ the set of all monomials, by $\polys$ the set of all polynomials over monomials in $\monomials$, and by $\obspolys = \{\obsop \in \polys \mid \obsop = \obsop^\dagger \} \subset \polys$ the set of all polynomials that are self-adjoint. Polynomials leave invariant the Schwartz space, so $\obspolys \subseteq \allOpsA$~\cite{Becnel15,Gieres00,Parviz24,Carcassi25}.
The proof of Proposition~\ref{prop:fock_proj_poly} is in the Appendix Section~\ref{appendix:prop:fock_proj_poly}.

\begin{proposition}[Fock-state projectors as polynomial limits] For all nonnegative integers $\nat \in \Nats$, we have $\toproj{\nat}{}=\lim_{\beta \rightarrow \infty} e^{-\beta\cdot (\nop - \nat\cdot\Id)^2},$ where the limit converges in trace-norm.
\label{prop:fock_proj_poly}
\end{proposition}

\paragraph{Density operators.}
An operator $\dm$ acting on $\hs$ is a density operator if (i) $\dm \geq \zeroOp$, (ii) $\dm = \dm^\dagger$, and (iii) $\Tr(\dm) \leq 1$. We denote the set of all density operators by $\alldm = \{\dm \in \allOpsAHs \mid \dm \geq \zeroOp \text{ and } \Tr(\dm) \leq 1\}$.
A {\em moment} of a quantum state is the expectation value of a monomial over canonical observables~\cite{CVbook}. In particular, the  {\em first moments} are the expectation values of all quadratures $\quadOps$; and the {\em second moments} the expectations of degree-2 monomials.
A quantum state $\dm$ is a (pure) {\em Gaussian state} if there exists a self-adjoint operator $\obsop$ that can be written as a polynomial over canonical observables of degree at most two such that $\dm = e^{-\im \obsop}\cdot \ket{0}\bra{0}\cdot e^{\im\obsop}$~\cite{Weedbrook12}.
Such states can be efficiently simulated and can be fully characterized by their first and second moments.
In this work, we are interested in a broad set of well-behaved density operators that guarantee finite expectation values on all observables in $\allOpsA$.
Thus, we consider the set $\alldmS = \{\dm \in \allOpsAHs \mid \dm \geq \zeroOp  \text{ and }\Tr(\dm) \leq 1 \text{ and } \forall \obsop \in \obspolys. \Tr(|\obsop\dm|) < \infty \} \subset \alldm$ of {\em Schwartz density operators}. Lemma~\ref{lemma:inv-dm}  shows Schwartz density operators preserve the dense domain $\denseDomain$ (see proof in Appendix~\ref{appendix:lemma:inv-dm}).
\begin{lemma}[Invariance of $\denseDomain$ for Schwartz density operators] For all vectors $\qvec \in \hs$, and Schwartz density operators $\dm \in \alldmS$; we have $\dm\ket{\qvec} \in \denseDomain$, and therefore,  $\alldmS \subset \allOpsA$.
\label{lemma:inv-dm}
\end{lemma}

\subsection{Continuous-variable quantum computing}
\label{subsec:cvqc}
We define the notion of (Schwartz) cq-states, (tempered) cq-observables, both of which are defined in general as an {\em operator-valued function (OVF)} $\cf: \cstates \rightarrow \allOpsHs$ from classical states to linear operators for $\hs$. This section concludes by defining 
 atomic instructions semantics.

\paragraph{Classical states.} Let $\allcv$ be a finite set of classical variables.
A {\em classical state} is a function $\cstate: \allcv \rightarrow \Nats$ that maps every classical variable to a nonnegative integer.
We denote by $\cstates$ the set of all classical states over $\allcv$.
Given a classical state $\cstate$, a classical variable $\cvar \in \allcv$ and a nonnegative integer $\nat \in \Nats$, we denote by $\cstate ' = \cwrite{\cstate}{\cvar}{\nat}$ the classical state obtained from $\cstate$ by  assigning $\nat$ to $\cvar$; i.e., $\cstate'(\cvar) = \nat$ and for all $\cvar' \neq \cvar$, we have $\cstate'(\cvar') = \cstate(\cvar')$.
We write $|\cstate| = \sum_{\cvar \in \allcv}\cstate(\cvar)$ for the sum of its  components.

\paragraph{Notation.} Given a set $S \subseteq (\cstates \times \allOpsHs)$, we sometimes denote an OVF $\cf$ in the explicit form $\cf = \bigoplus_{(\cstate, \op)\in S} \tohstate{\cstate}{\op}$.
The support of $\cf$ is the set $\Supp{\cf} = \{\cstate \mid \cf(\cstate) \neq \zeroOp\}$ of classical states whose image is not the zero operator.
Given two OVFs $\cf, \cf'$, we have $(\cf\otimes\cf')(\cstate) = \cf(\cstate) \cdot \cf'(\cstate)$ for all classical states $\cstate \in \cstates$.
For a linear operator $\op \in \allOpsHs$, we write $\cf' = \op\cdot\cf$ where  we define $\cf'(\cstate) = \op\cdot\cf(\cstate)$ for all classical states $\cstate\in\cstates$. 
We have $\cf=\cf'$ if $\cf(\cstate) \deq \cf'(\cstate)$ for all  $\cstate \in \cstates$.
The trace of $\cf$ is given by $\Tr(\cf) = \sum_{\cstate\in\cstates}\Tr(\cf(\cstate))$ and $\Tr(|\cf|) = \sum_{\cstate\in\cstates}\Tr(|\cf(\cstate)|)$.

\paragraph{Quantum states.} 
Let $\allqv$  be a finite set of quantum variables.
A {\em quantum state} is a density operator in $\alldm$.
The subsystem associated with a quantum variable $\qvar \in \allqv$ is denoted by $\hs_\qvar \simeq L^2(\Reals)$, and
the combined system is $\hs_{\allqv} = \bigotimes_{\qvar \in \allqv}\hs_\qvar \simeq \hs \simeq L^2(\mReals{\numModes})$. 
We adopt the convention from the quantum computing literature when operators apply to subsystems~\cite{qc_bible}:
given an operator $\op$ for  $\hs_\qvar$, we denote by $\op_\qvar$ the linear operator that applies $\op$ to $\qvar$  and the identity to the remaining variables. The total number operator is  $\nop_{\allqv} = \sum_{\qvar \in \allqv} \nop_{\qvar}$.
A vector $\ket{\qvec}$ in $\hs_{\qvar}$ may be denoted by $\ket{\qvec}_{\qvar}$, unless clear from context. Given an operator $\op$ on $\hs_{\allqv}$, we have $\bra{\qvec}_{\qvar}\cdot\op\cdot\ket{\qvec}_\qvar =  (\bra{\qvec}_{\qvar}\otimes\Id)\cdot\op\cdot(\ket{\qvec}_\qvar\otimes\Id)$.

\paragraph{Classical-quantum states (cq-states).}
A {\em cq-state} is an OVF $\hstate: \cstates \rightarrow \alldm$
such that $\sum_{\cstate \in \cstates} \hstate(\cstate) \in \alldm$.
We say that $\hstate$ is a {\em Schwartz cq-state} if is rapidly decreasing with respect to the classical state, namely, for all $k \in \Nats$,
\begin{equation}
    \sum_{\cstate \in \cstates}(1+|\cstate|)^k\cdot\Tr((\nop_{\allqv}+\Id)^k\cdot\hstate(\cstate)) < \infty.
    \label{eq:schwartz_dm}
\end{equation}
Equation~\ref{eq:schwartz_dm} implies that $\hstate(\cstate)$ is itself a Schwartz density operator since $\Tr((\nop_{\allqv}+\Id)^k\cdot \hstate(\cstate)) < \infty$ for all $k \in \Nats$ and $\cstate \in \cstates$.
We denote the set of all cq-states by $\hstates$, and the set of all Schwartz cq-states by $\hstatesS$.

\paragraph{Classical-quantum observables (cq-observables).}
A {\em cq-observable} is an OVF $\hobs: \cstates \rightarrow \allOpsAHs$ that maps classical states to observables on $\hs$.
We say $\hobs$ is a {\em tempered cq-observable} if $\hobs(\cstate) \in \allOpsA$ is an observable  that preserves the Schwartz space for all $\cstate \in \cstates$, and the growth of $\hobs$ is
polynomially bounded in both the classical and
quantum states.
Formally, for every $\nat \in \Nats$ there are $C_\nat, k_\nat, b_\nat \in \Nats$ such that for all $\cstate \in \cstates$ and $\qvec \in \denseDomain$,
\begin{equation}
\norm{(\nop_{\allqv}+\Id)^\nat\cdot\hobs(\cstate)\cdot\ket{\qvec}} \leq C_\nat\cdot(1+|\cstate|)^{k_\nat}\cdot\norm{(\nop_{\allqv}+\Id)^{b_\nat}\cdot\ket{\qvec}}.
    \label{eq:well_behaved}
\end{equation}
This condition ensures that the expectation value
of a tempered cq-observable on a Schwartz cq-state is finite  
(Lemma~\ref{lemma:finite_exp_stemp};
proof in Appendix~\ref{appendix:lemma:finite_exp_stemp}).
We denote the sets of all cq-observables and
tempered cq-observables by $\allhobs$ and
$\allhobsS$, respectively.

\begin{lemma}[Expectation values of tempered cq-observables on Schwartz cq-states] For all Schwartz cq-states $\hstate \in \hstatesS$ and tempered cq-observables $\hobs \in \allhobsS$, holds $\Tr(|\hobs\otimes\hstate|) < \infty$.
\label{lemma:finite_exp_stemp}
\end{lemma}

\paragraph{Atomic instructions.}
We consider three broad kinds of atomic instructions: (i)~unitary instructions, (ii)~measurements in the Fock basis, and (iii)~reset instructions.
Unitary and reset instructions leave the classical state unchanged, but  measurements record outcomes in the classical state.
Formally, an {\em atomic instruction} is a countable family $\kset = \{(\op_k, u_k)\}_{k \in I}$ of pairs, where each $\op_k$ is a bounded operator on $\hs_{\allqv}$ and each $u_k: \cstates \rightarrow \cstates$ is a {\em classical update}, such that the Kraus operators satisfy $\sum_{k \in I} \op_k^\dagger \op_k = \Id$ (strong equality).
The {\em forward semantics} $\sem{\kset}: \hstates \rightarrow \hstates$ maps cq-states to cq-states, and the {\em backward semantics} $\semD{\kset}: \allhobs \rightarrow \allhobs$ maps cq-observables to cq-observables. Formally
$\sem{\kset}(\hstate)(\cstate') \deq \sum_{\cstate \in \Supp{\hstate}}~\sum_{k \in I,\, u_k(\cstate) = \cstate'} \op_k\cdot\hstate(\cstate)\cdot\op_k^\dagger$ and $\semD{\kset}(\hobs)(\cstate) \deq \sum_{k \in I}\op_k^\dagger\cdot\hobs(u_k(\cstate))\cdot\op_k$.
The concrete Kraus operators of all considered atomic instructions are defined in Table~\ref{tab:atomic_kraus}, where $\unitaries = \{\DispI{\qvar}{\real}{\real'}, \SqueezeI{\qvar}{\real}, \RotI{\qvar}{\real}, \CubicPI{\qvar}{\real}, \BSI{\qvar}{\qvar'}{\real}{\real'} \mid \qvar, \qvar' \in \allqv \text{ and } \real, \real' \in \Reals\}$ is the set of all considered {\em atomic unitary instructions}. Each unitary atomic instruction denotes a unitary operator acting on $\hs_{\allqv}$: for real numbers $\real, \real' \in \Reals$ and distinct quantum variables $\qvar, \qvar' \in\allqv$:
\begin{equation}
    \begin{array}{lll}
         \DispI{\qvar}{\real}{\real'} = e^{(-\im\cdot\sqrt{2}\cdot(\real\cdot\POp_\qvar - \real'\cdot\X_\qvar))} &
        \SqueezeI{\qvar}{\real} = e^{\left(\real\cdot\left(\dop_\qvar^2 - {(\cop_\qvar)}^2\right)/2\right)} &
        \RotI{\qvar}{\real} = e^{(\im\cdot\real\cdot\cop_\qvar\cdot\dop_\qvar)} \\
        \CubicPI{\qvar}{\real} = e^{\left(\im\cdot\real\cdot\X^3_\qvar/3\right)} & \multicolumn{2}{c}{\BSI{\qvar}{\qvar'}{\real}{\real'} = e^{(\real\cdot(e^{\im\cdot\real'}\cdot(\dop_{\qvar}\cdot\cop_{\qvar'})- e^{-\im\cdot\real'}\cdot(\cop_{\qvar}\cdot\dop_{\qvar'})))}}.
    \end{array}
\end{equation}
We denote the set of all atomic instructions by $\instructions = \unitaries \cup \{(\MeasI{\cvar}{\qvar}),~\ResetI(\qvar)\}$.
These unitaries correspond to a universal gate set and are the standard building blocks of continuous-variable quantum programs~\cite{Weedbrook12,CVbook,Killoran19,brask22,Kok10}, according to 
the standard universality criterion for CQC~\cite{Lloyd99}, which requires at least one gate of polynomial degree higher than 2~\cite{Braunstein05}, such as the cubic-phase gate $\CubicPI{\qvar}{\real}$~\cite{Weedbrook12,miwa09}.
The syntactic substitutions in Section~\ref{sec:hoare} for unitary instructions exploit the fact that unitary conjugation is an algebra automorphism; so the backward semantics of a polynomial is obtained by substituting each ladder operator by its conjugate~\cite{Killoran19}:
{\small\begin{equation}
    \begin{array}{cc}
        \DispI{\qvar}{\real}{\real'}^\dagger\cdot\dop_\qvar\cdot\DispI{\qvar}{\real}{\real'} = \dop_\qvar + (\real + \im\cdot\real')\cdot\Id &
        \SqueezeI{\qvar}{\real}^\dagger\cdot\dop_\qvar\cdot\SqueezeI{\qvar}{\real} = \cosh(\real)\cdot\dop_\qvar - \sinh(\real)\cdot\cop_\qvar \\[2pt]
        \RotI{\qvar}{\real}^\dagger\cdot\dop_\qvar\cdot\RotI{\qvar}{\real} = e^{\im\cdot\real}\cdot\dop_\qvar &
        \CubicPI{\qvar}{\real}^\dagger\cdot\dop_\qvar\cdot\CubicPI{\qvar}{\real} = \dop_\qvar + \frac{\im\cdot\real}{2\sqrt{2}}(\dop_\qvar + \cop_\qvar)^2 \\[2pt]
        \multicolumn{2}{c}{\BSI{\qvar}{\qvar'}{\real}{\real'}^\dagger\cdot\dop_{\qvar}\cdot\BSI{\qvar}{\qvar'}{\real}{\real'} = \cos(\real)\cdot\dop_{\qvar} - e^{-\im\cdot\real'}\cdot\sin(\real)\cdot\dop_{\qvar'}} \\[2pt]
        \multicolumn{2}{c}{\BSI{\qvar}{\qvar'}{\real}{\real'}^\dagger\cdot\dop_{\qvar'}\cdot\BSI{\qvar}{\qvar'}{\real}{\real'} = \cos(\real)\cdot\dop_{\qvar'} + e^{\im\cdot\real'}\cdot\sin(\real)\cdot\dop_{\qvar}}.
    \end{array}
    \label{eq:dual_unitaries}
\end{equation}}
\begin{table}[!t]
    \centering
    \begin{tabular}{lccc}
        \hline
        atomic instruction & index set $I$ & Kraus operator $\op_k$ & classical update $u_k(\cstate)$ \\
        \hline
        unitary $U \in \unitaries$ & $\{\ast\}$ & $U$ & $\cstate$ \\
        $\MeasI{\cvar}{\qvar}$ & $\Nats$ & $\toproj{\nat}{\qvar}$ & $\cwrite{\cstate}{\cvar}{\nat}$ \\
        $\ResetI(\qvar)$ & $\Nats$ & $\ket{0}_\qvar\bra{\nat}$ & $\cstate$ \\
        \hline
    \end{tabular}
    \Description{Kraus operators and classical updates of atomic instructions.}
    \caption{Kraus operators and classical updates of atomic instructions, where $\cvar \in \allcv$ and $\qvar \in \allqv$.}
    \label{tab:atomic_kraus}
\end{table}

%% file: sections/_5_pl.tex
\section{Continuous-variable Quantum Programming Language (\myPL)}
\label{sec:pl}
In this section, we introduce $\myPL$, a universal programming language for continuous-variable quantum systems, together with its admissible fragment.

\paragraph{Expressions.} Our programming language relies on real expressions for parameters of unitary instructions, as well as Boolean expressions for control-flow.
The syntax and semantics of real and Boolean  expressions is shown in Figure~\ref{fig:bexprs}. A {\em real expression} maps a classical state to a real value, and a {\em Boolean expression} denotes sets of classical states satisfying some inequalities. 
We chose the real-expression syntax to be closed under the dual maps  (Equation~\ref{eq:dual_unitaries}), i.e., under the trigonometric and hyperbolic identities they produce.
We say that a real expression $\rexpr$ is {\em polynomially bounded} if there exist $C, m \in \Nats$ such that $|\sem{\rexpr}(\cstate)| \leq C\cdot(1+|\cstate|)^m$ for all classical states $\cstate \in \cstates$.
Given a Boolean expression $\bexpr$ and an OVF $\cf$, we denote by $\hstateAt{\cf}{\bexpr} = \bigoplus_{\cstate \in \sem{\bexpr}} \tohstate{\cstate}{\cf(\cstate)}$ the OVF that maps every classical state that does not satisfy $\bexpr$ to $\zeroOp$ and equals $\cf$ otherwise.
\begin{figure}[!t]
    \centering
    \scriptsize
    \[
    \begin{array}{lcl}
        \text{Real expression} \quad \rexpr~\defeq~\real \mid \cvar \mid e^{\rexpr} \mid \cos(\rexpr) \mid \rexpr_0 + \rexpr_1 \mid \rexpr_0 \cdot \rexpr_1 & \quad\quad\quad &
        \text{Boolean expression} \quad \bexpr ~\defeq~\rexpr_0 \geq \rexpr_1 \mid \neg \bexpr \mid \bexpr_0 \land \bexpr_1 \\
    \end{array}
    \]
    {\renewcommand{\arraystretch}{1.5}\begin{tabular}{|c|llll|}
    \hline
         signature & \multicolumn{4}{c|}{semantics} \\
    \hline
    {
        \multirow{2}{*}{$\sem{\rexpr}: \cstates \rightarrow \Reals $}}
        & $\sem{\cvar}(\cstate) = \cstate(\cvar)$ 
        & \multicolumn{1}{c}{$\sem{\real}(\cstate) = \real$}
        & $\sem{e^{\rexpr}}(\cstate) = e^{\sem{\rexpr}(\cstate)}$
        & $\sem{\cos(\rexpr)}(\cstate) = \cos(\sem{\rexpr}(\cstate))$
        \\
        & \multicolumn{2}{c}{$\sem{\rexpr_0 + \rexpr_1}(\cstate) = \sem{\rexpr_0}(\cstate) + \sem{\rexpr_1}(\cstate)$}
        & \multicolumn{2}{l|}{$\sem{\rexpr_0 \cdot \rexpr_1}(\cstate) = \sem{\rexpr_0}(\cstate) \cdot \sem{\rexpr_1}(\cstate)$ }
        \\[2pt]
    \hline
    {$\sem{\bexpr} \subseteq \cstates $}
     & \multicolumn{2}{l}{$\sem{\rexpr_0 \geq \rexpr_1} = \{\cstate \in \cstates \mid \sem{\rexpr_0}(\cstate) \geq \sem{\rexpr_1}(\cstate)\}$}
     &\multicolumn{1}{c}{$\sem{\neg\bexpr} = \cstates \setminus\sem{\bexpr}$}&  
     \multicolumn{1}{c|}{$\sem{\bexpr_0 \land \bexpr_1} = \sem{\bexpr_0} \cap \sem{\bexpr_1}$} 
      \\[2pt]
      \hline
    \end{tabular}}
    \Description{Syntax  and semantics of Boolean expressions}
    \caption{Syntax  and semantics of  real and Boolean expressions, where $\cvar \in \allcv$ is a classical variable.} 
    \label{fig:bexprs}
\end{figure}

\paragraph{Syntax and semantics of $\myPL$.} 
An atomic program is an atomic instruction whose parameters may be real expressions, evaluated in the current classical state.
As discussed in Section~\ref{subsec:cvqc}, $\unitaries$ is a universal gate set, so \myPL~is a universal programming language for CQC.
Similar to atomic instructions, we define both forward and backward semantics in Figure~\ref{fig:pl}. For atomic programs, we define both semantics with respect to an OVF $\cf \in \hstates \cup \allhobs$.
We intentionally omit an ideal atomic program to measure position because such an instruction is ill-defined without assuming a finite measurement resolution~\cite{Buck21,Ballentine14}. Experimentally, quantum computers can approximate quadrature measurements using measurements in the Fock basis~\cite{Lloyd99,Braunstein05,Leonhardt95}, as shown in our experiments (Section~\ref{sec:exp_fock}).

\paragraph{Admissible \myPL~programs.} In infinite dimensions, unbounded loops can lead to divergent expectation values even in non-universal (e.g., Gaussian) languages.
Similarly, unitary parameters that grow too fast in the classical state break the polynomial growth and decay conditions of tempered cq-observables and Schwartz cq-states; e.g., $\DispI{\qvar}{e^{\cvar}}{0}$ after measuring $\cvar$.
For this reason, our logic restricts to the set $\progs$ of admissible programs.
Intuitively, a $\myPL$~program $\prog$ is admissible if every loop in $\prog$ halts after at most $k$ iterations from every cq-state, and the coefficients of the dual of every unitary (Equation~\ref{eq:dual_unitaries}) grow at most polynomially in the classical state.
A \myPL~atomic program $\aprog$ is {\em admissible} iff one of the following holds:
(1)~$\aprog = \DispI{\qvar}{\rexpr_0}{\rexpr_1}$ and both $\rexpr_0$ and $\rexpr_1$ are polynomially bounded; 
(2)~$\aprog = \CubicPI{\qvar}{\rexpr}$ and $\rexpr$ is polynomially bounded; (3)~$\aprog = \SqueezeI{\qvar}{\rexpr}$ and $e^{|\rexpr|}$ is polynomially bounded; or
(4)~$\aprog \in \instructions \cup \{\BSI{\qvar_0}{\qvar_1}{\rexpr_0}{\rexpr_1}, \RotI{\qvar}{\rexpr}\}$.
A rotation $\RotI{\qvar}{\rexpr}$ and beam splitter $\BSI{\qvar_0}{\qvar_1}{\rexpr_0}{\rexpr_1}$ instructions do not have a condition, since their duals  introduce only bounded functions in their parameters.

In general, a \myPL~program $\prog$ is {\em admissible} iff
one of the following holds: (1) $\prog$ is an admissible atomic program; (2)~$\prog = \prog_0;\prog_1$ or $\prog = \IfElse{\bexpr}{\prog_0}{\prog_1}$, and $\prog_0$ and $\prog_1$ are admissible; or (3)~$\prog = \While{\bexpr}{\prog'}$ and $\prog'$ is admissible, and there exists $k \in \Nats$ such that $\Phi_{\bexpr,\prog'}^{k}(\hstateAt{\hobs_\Id}{\bexpr}) = \zeroOp$;
where $\Phi_{\bexpr,\prog'}(\hobs) = \hstateAt{\semD{\prog'}(\hobs)}{\bexpr}$ is the backward semantics of one guarded iteration, and $\hobs_\Id$ the tempered cq-observable that maps every classical state to $\Id$. So intuitively, $\Phi_{\bexpr,\prog'}^{j}(\hstateAt{\hobs_\Id}{\bexpr})$ is the observable whose expectation is the probability that the loop starts its $(j+1)$-th iteration.
For admissible loops with bound $k$, both sequences stabilize after $k{+}1$ steps, so the limits exist.
We denote the set of all admissible \myPL~programs by $\progs$.
Theorem~\ref{thm:schwartz_invariance} shows duality of forward and backward semantics for admissible programs on Schwartz cq-states and tempered cq-observables, respectively (see proof in Appendix Section~\ref{appendix:thm:schwartz_invariance}).

\begin{theorem}[Invariance and duality under admissible \myPL-program semantics]
\label{thm:schwartz_invariance}
For all admissible \myPL~programs $\prog \in \progs$, Schwartz cq-states $\hstate \in \hstatesS$, and tempered cq-observables $\hobs\in\allhobsS$; we have (1)~$\sem{\prog}(\hstate) \in \hstatesS$, (2)~$\semD{\prog}(\hobs) \in \allhobsS$, and (3)~$\Tr(\semD{\prog}(\hobs)\otimes \hstate) = \Tr(\hobs\otimes\sem{\prog}(\hstate)) < \infty.$
\end{theorem}

\begin{figure}[!t]
\scriptsize
\[
\begin{array}{rcl}
\text{Atomic program} \quad \aprog & \defeq & \instruction \mid  \DispI{\qvar}{\rexpr_0}{\rexpr_1} \mid
    \SqueezeI{\qvar}{\rexpr} \mid
    \RotI{\qvar}{\rexpr} \mid \CubicPI{\qvar}{\rexpr} \mid \BSI{\qvar_0}{\qvar_1}{\rexpr_0}{\rexpr_1} \\
\text{Program} \quad \prog & \defeq & \aprog \mid \prog_0;\prog_1 \mid \IfElse{\bexpr}{\prog_0}{\prog_1} \mid \While{\bexpr}{\prog}
\end{array}
\]
    \[
    \begin{array}{ll}
        \multicolumn{1}{l}{
        \sem{\DispI{\qvar}{\rexpr_0}{\rexpr_1}}^{(*)}(\cf) 
        = \bigoplus_{\cstate \in \cstates}\sem{\DispI{\qvar}{\sem{\rexpr_0}(\cstate)}{\sem{\rexpr_1}(\cstate)}}^{(*)}(\tohstate{\cstate}{\cf(\cstate)})}
    &
    % squeeze
    \multicolumn{1}{l}{\sem{\SqueezeI{\qvar}{\rexpr}}^{(*)}(\cf) 
        = \bigoplus_{\cstate \in \cstates} \sem{\SqueezeI{\qvar}{\sem{\rexpr}(\cstate)}}^{(*)}(\tohstate{\cstate}{\cf(\cstate)})}
    \\[2pt]
    % cubic phase
    \multicolumn{1}{l}{\sem{\RotI{\qvar}{\rexpr}}^{(*)}(\cf) 
        = \bigoplus_{\cstate \in \cstates}\sem{\RotI{\qvar}{\sem{\rexpr}(\cstate)}}^{(*)}(\tohstate{\cstate}{\cf(\cstate)})}
    &
    \multicolumn{1}{l}{\sem{\CubicPI{\qvar}{\rexpr}}^{(*)}(\cf) = \bigoplus_{\cstate \in \cstates}\sem{\CubicPI{\qvar}{\sem{\rexpr}(\cstate)}}^{(*)}(\tohstate{\cstate}{\cf(\cstate)})}
    \\[2pt]
    
    % beam splitter
     \multicolumn{2}{l}{\sem{\BSI{\qvar_0}{\qvar_1}{\rexpr_0}{\rexpr_1}}^{(*)}(\cf) = \bigoplus_{\cstate \in \cstates}\sem{\BSI{\qvar_0}{\qvar_1}{\sem{\rexpr_0}(\cstate)}{\sem{\rexpr_1}(\cstate)}}^{(*)}(\tohstate{\cstate}{\cf(\cstate)}) }
    \\[2pt]
   { \sem{\IfElse{\bexpr}{\prog_0}{\prog_1}}(\hstate) = \sem{\prog_0}(\hstateAt{\hstate}{\bexpr})\oplus \sem{\prog_1}(\hstateAt{\hstate}{\neg\bexpr})}&
   \sem{\prog_0; \prog_1}(\hstate) = \sem{\prog_1}(\sem{\prog_0}(\hstate)) \\[2pt]
    { \semD{\IfElse{\bexpr}{\prog_0}{\prog_1}}(\hobs) = \hstateAt{(\semD{\prog_0}(\hobs))}{\bexpr} \oplus \hstateAt{(\semD{\prog_1}(\hobs))}{\neg\bexpr}} & \semD{\prog_0; \prog_1}(\hobs) = \semD{\prog_0}(\semD{\prog_1}(\hobs))\\[2pt]
    \multicolumn{2}{l}{\sem{\While{\bexpr}{\prog}}(\hstate) = \lim_{n\rightarrow\infty} F^{n}(\mathbf{0})(\hstate),
    \quad \text{where } 
    F(X)(\hstate) 
      =  X\left(\sem{\prog}\left(\hstateAt{\hstate}{\bexpr}\right)\right) 
        \oplus \hstateAt{\hstate}{\neg\bexpr}}
        \\[2pt]
    \multicolumn{2}{l}{\semD{\While{\bexpr}{\prog}}(\hobs) = \lim_{n\rightarrow\infty} (F^{*})^{n}(\mathbf{0})(\hobs),
    \quad \text{where } 
    F^*(X)(\hobs) 
      =  \hstateAt{\semD{\prog}\left(X\left(\hobs\right)\right)}{\bexpr}
        \oplus \hstateAt{\hobs}{\neg\bexpr}}
    \end{array}
    \]
    \Description{
    On the top is the syntax of $\myPL$. The bottom defines both the forward semantics (a function $\sem{\cdot}: \hstates \rightarrow \hstates$ on cq-states) and the backward semantics (a function $\semD{\cdot}: \allhobs \rightarrow \allhobs$ on cq-observables).
    }
    \caption{On the top is the syntax of $\myPL$. The bottom defines both the forward semantics (a function $\sem{\cdot}: \hstates \rightarrow \hstates$ on cq-states) and the backward semantics (a function $\semD{\cdot}: \allhobs \rightarrow \allhobs$ on cq-observables).}
    \label{fig:pl}
\end{figure}

%% file: sections/_6_hoare.tex
\section{Continuous-variable Quantum Hoare logic (\myHL)}
\label{sec:hoare}
This section introduces \myHL, a Hoare logic for continuous-variable quantum programs.
First, we introduce the syntax and semantics of assertions, and define syntactic substitutions  used to derive weakest preconditions.
Then we present a relatively complete proof system for partial correctness.

\paragraph{Syntax and semantics of assertions.}
\begin{figure}[!t]
    \centering
    {       \scriptsize
    \begin{tabular}{rcl}
            \texttt{Polynomial term} &   \(\polyterm :=\)   & $\rexpr\cdot\monomial \mid \polyterm_0 + \polyterm_1$     \\
           \texttt{Quantum term} &   \(\qterm :=\)   & \(\polyterm \mid \kerterm{\polyterm} \mid \bindterm{\cvar}{\qterm} \mid \vacterm{\anc}{\qterm} \mid \guardterm{\bexpr}{\qterm} \mid \qterm_0 + \qterm_1 \mid \qterm_0 \cdot \qterm_1\) \\
              \texttt{Assertion} & \(\assertion :=\) & \(\qterm_0~\geq~\qterm_1 \mid \neg \assertion \mid \assertion \land \assertion\)
    \end{tabular}}
    {\scriptsize
        \renewcommand{\arraystretch}{1.5}\begin{tabular}{|c|llll|}
    \hline
         signature & \multicolumn{4}{c|}{semantics} \\
    \hline
        {\multirow{1}{*}{$\sem{\polyterm} \in \allhobs$}}
        & \multicolumn{1}{l}{$\sem{\rexpr\cdot \monomial} = \bigoplus_{\cstate \in \cstates} \tohstate{\cstate}{\sem{\rexpr}(\cstate) \cdot \monomial}$} 
        &\multicolumn{3}{l|}{$\sem{\polyterm_0 + \polyterm_1} =  \sem{\polyterm_0} \oplus \sem{\polyterm_1}$} 
        \\[2pt]
        \hline
        {\multirow{3}{*}{$\sem{\qterm} \in \allhobs$}}
        &\multicolumn{1}{l}{$\sem{\kerterm{\polyterm}} = \bigoplus_{\cstate\in \cstate} \tohstate{\cstate}{\lim_{\beta \rightarrow \infty} e^{-\beta \cdot \sem{\polyterm}(\cstate)}}$} &\multicolumn{2}{l}{$\sem{\bindterm{\cvar}{\qterm}} = \sum_{\nat=0}^{\infty} \sem{\csubs{\qterm}{\cvar}{\nat}}$}&
        \\[2pt]
     & {$\sem{\vacterm{\anc}{\qterm}} = \bra{0}_\anc\sem{\qterm}\ket{0}_\anc$} &  
     \multicolumn{2}{l}{$\sem{\guardterm{\bexpr}{\qterm}} = \bigoplus_{\cstate \in \sem{\bexpr}}\tohstate{\cstate}{\sem{\qterm}(\cstate)}$}&
      \\[2pt]
      &$\sem{\qterm_0 + \qterm_1} = \sem{\qterm_0} \oplus \sem{\qterm_1}$
      & $\sem{\qterm_0 \cdot \qterm_1} = \bigoplus_{\cstate\in\cstates}\tohstate{\cstate}{\sem{\qterm_0}(\cstate) \cdot \sem{\qterm_1}(\cstate)}$ & &\\[2pt]
      \hline
      \multirow{2}{*}{
        $\sem{\assertion}\subseteq \hstatesS$}
        &\multicolumn{4}{l|}{$\sem{\qterm_0 \geq \qterm_1} = \{\hstate \in \hstatesS \mid \sum_{\cstate \in \cstates} \Tr((\sem{\qterm_0} \otimes \hstate)(\cstate)) \geq \sum_{\cstate \in \cstates} \Tr((\sem{\qterm_1}\otimes \hstate)(\cstate))\}$}\\
        & {$\sem{\assertion_1~\land~\assertion_2} = \sem{\assertion_1} \cap \sem{\assertion_2}$} & {$\sem{\neg \assertion} = \hstatesS \setminus \sem{\assertion}$} &&
        \\[2pt]
    \hline
    \end{tabular}}
    \Description{Syntax and semantics of quantum terms and assertions.
        A quantum term denotes a cq-observable and assertions denote sets of Schwartz cq-states.
        Here $\monomial \in \polys$ is a normally-ordered monomial over ladder operators, $\rexpr$ is a real expression, $\bexpr$ a Boolean expression, $\cstate$ a classical state, $\cvar$ a classical variable, and $\anc$ an ancilla variable.}
    \caption{
        Syntax and semantics of quantum terms and assertions.
        A quantum term denotes a cq-observable and assertions denote sets of Schwartz cq-states.
        Here $\monomial \in \polys$ is a normally-ordered monomial over ladder operators, $\rexpr$ is a real expression, $\bexpr$ a Boolean expression, $\cstate$ a classical state, $\cvar$ a classical variable, and $\anc$ an ancilla variable.}
    \label{fig:assertions}
\end{figure}

The syntax of quantum terms and assertions is shown in
Figure~\ref{fig:assertions}. A {\em quantum term}  
\(\qterm\) denotes a cq-observable.
A quantum term $\qterm$ may be specified using a (classical) guard $\qterm = \guardterm{\bexpr}{\qterm'}$ which denotes a cq-observable that is equal to the observable denoted by $\qterm'$ but only for those classical states that satisfy $\bexpr$. For those classical states that do not satisfy $\bexpr$, their image is the zero operator.
If a guard is not specified (i.e., implicit), then the quantum term denotes a cq-observable that maps every classical state to the same observable. 
For example, a polynomial term  $\polyterm = \poly \in \obspolys$ denotes the cq-observable $\hobs$ defined by $\hobs(\cstate) = \poly $ for all classical states $\cstate \in \cstates$.
By Proposition~\ref{prop:fock_proj_poly}, Fock-state projectors are specified using a quantum term of the form $\kerterm{\polyterm}$.
Intuitively, the quantum term of the form $\sum_{\cvar} \qterm$ arises in measurements and denotes an infinite sum of cq-observables.  
Quantum terms may mention ancilla variables wherever they mention quantum variables, and the quantum term \(\vacterm{\anc}{\qterm}\) binds \(\anc\) in \(\qterm\), in which $\anc$ is an ancilla variable drawn from a countably infinite set \(\ancv\) of {\em ancilla variables}, disjoint from \(\allqv\), each associated with a qumode \(\hs_\anc \simeq L^2(\Reals)\). 
Such terms are used in the syntactic substitution that computes the backward  semantics of the reset instruction.
Intuitively, the substitution can be interpreted as replacing $\qvar$ with a qumode in the vacuum.

An {\em assertion} \(\pre\) is a Boolean combination of
comparisons \(\qterm_0 \geq \qterm_1\) between the expectation values of the two cq-observables represented using quantum terms; and an assertion denotes the set of Schwartz cq-states whose expectation values satisfy the stated inequalities.
We may write $\qterm_0 = \qterm_1$ and $\pre_0 \lor \pre_1$ for $\qterm_0 \geq \qterm_1 \land \qterm_1 \geq \qterm_0$, and $\neg(\neg\pre_0 \land \neg\pre_1)$, respectively.
Given a cq-state $\hstate$ we write $\hstate \vDash \assertion$ if $\hstate \in \sem{\assertion}$.
Since the syntax does not restrict quantum terms to tempered cq-observables, we define the set $\assertions$ of {\em well-formed assertions}.
Formally, an assertion $\pre\in\assertions$ iff one of the following hold: (1)~$\pre = \qterm_0 \geq \qterm_1$ and $\sem{\qterm_0}, \sem{\qterm_1} \in \allhobsS$ are tempered cq-observables; (2)~$\pre = \neg\pre'$ and $\pre' \in \assertions$; or (3)~$\pre = \pre_0 \land \pre_1$ and $\pre_0,\pre_1 \in \assertions$.

\begin{definition}[Validity]
    For two well-formed assertions \(\pre, \post \in \assertions\) and all \myPL~admissible programs \(\prog \in \progs\), the triple \(\{\pre\}~\prog~\{\post\}\) is valid if for all Schwartz cq-states \(\hstate \in \hstatesS\), if \(\hstate \vDash \pre\), then
    \(\sem{\prog}(\hstate) \vDash \post\).
\end{definition}
For two well-formed assertions \(\pre, \pre' \in \assertions\), we say that $\pre$ entails $\pre'$, written $\pre \Rightarrow \pre'$, if $\sem{\pre} \subseteq \sem{\pre'}$.
This notion of entailment is semantic rather than syntactic, but it can be expressed syntactically as the validity of \(\neg \pre \lor \pre'\). 
By the invariance guarantees of admissible programs semantics (Theorem~\ref{thm:schwartz_invariance}) and finite expectation values of tempered cq-observables on Schwartz cq-states (Lemma~\ref{lemma:finite_exp_stemp}), both validity and entailment of well-formed assertions are well-defined.

\paragraph{Syntactic substitutions.}
We define the function \(\atomsubs{\aprog}{\qterm}\) that maps an atomic program \(\aprog\) and a quantum term \(\qterm\) to a quantum term.
Lemma~\ref{lemma:dual_map} shows that substitutions compute the cq-observable that results from the backward evolution of $\qterm$ by $\aprog$ (see proof in Appendix Section~\ref{appendix:lemma:dual_map}).
Formally, \(\atomsubs{\aprog}{\qterm}\) is defined for all atomic programs \(\aprog \in \aprogs\) as follows:
\begin{itemize}
    \item if $\qterm = \qterm_0 + \qterm_1$, then \(\atomsubs{\aprog}{\qterm} = \atomsubs{\aprog}{\qterm_0} + \atomsubs{\aprog}{\qterm_1}\),
    \item if $\aprog$ is a unitary atomic program, then $\atomsubs{\aprog}{\qterm}$ is the quantum term that results from replacing every $\dop_{\qvar'}$ by $\semD{\aprog}(\dop_{\qvar'})$ and every $\cop_{\qvar'}$ by $(\semD{\aprog}(\dop_{\qvar'}))^\dagger$, for $\qvar' \in \allqv$, after renaming the bound classical variables of $\qterm$ apart from those in the parameters of $\aprog$.
    \item If \(\aprog = (\MeasI{\cvar}{\qvar})\):
    \begin{itemize}
        \item For $\qterm = \rexpr\cdot\monomial_{\mvals{d},\mvals{c}}$ where $\monomial_{\mvals{d},\mvals{c}}$ is a monomial; if $\mvals{d}_{\qvar} = \mvals{c}_{\qvar}$, then \(\atomsubs{\aprog}{\qterm} = \bindterm{\cvar}{\kerterm{(\nop_{\qvar}-\cvar)^2}\cdot\qterm}\). However, if $\cvar$ does not occur in $\rexpr$, this is equivalent to  $\atomsubs{\aprog}{\qterm} = \qterm$. 
        Otherwise, if $\mvals{d}_{\qvar} \neq \mvals{c}_{\qvar}$, then \(\atomsubs{\aprog}{\qterm} = \zeroOp \). 
        \item For all other cases, $\atomsubs{\aprog}{\qterm} = \bindterm{\cvar}{\kerterm{(\nop_{\qvar}-\cvar)^2}\cdot\qterm\cdot\kerterm{(\nop_{\qvar}-\cvar)^2}}$.
    \end{itemize}
    \item if \(\aprog = \ResetI(\qvar)\): 
    \begin{itemize}
        \item if $\qterm = \rexpr\cdot\monomial_{\mvals{d},\mvals{c}}$ where $\monomial_{\mvals{d},\mvals{c}}$ is a monomial; then if $\mvals{d}_{\qvar} = \mvals{c}_{\qvar}$, then \(\atomsubs{\aprog}{\qterm}= \rexpr\cdot(\prod_{i=0}^{\nat-1}\sqrt{(i+1)\cdot(\nat - i)})\cdot \prod_{\qvar' \in (\allqv\setminus \{\qvar\})}\dop_{\qvar'}^{\mvals{d}_{\qvar'}}\cdot(\cop_{\qvar'})^{\mvals{c}_{\qvar'}}\), where $\nat = \mvals{d}_{\qvar}$. Otherwise, $\mvals{d}_{\qvar} \neq \mvals{c}_{\qvar}$ and \(\atomsubs{\aprog}{\qterm} = \zeroOp \).
        \item If $\qterm \in \{~\bindterm{\cvar}{\qterm'},~\vacterm{\anc}{\qterm'} ,~\guardterm{\bexpr}{\qterm'}~\}$, then $\atomsubs{\aprog}{\qterm} =\subs{\qterm}{\qterm'}{ \atomsubs{\aprog}{\qterm'}}$ returns the quantum term that substitutes the inner quantum term $\qterm'$ by $\atomsubs{\aprog}{\qterm'}$. 
        \item If $\qterm = \kerterm{\polyterm}$ or $\qterm = \qterm_0\cdot \qterm_1$, then  \(\atomsubs{\aprog}{\qterm} = \vacterm{\anc}{\csubs{\qterm}{\qvar}{\anc}}\), where \(\anc \in \ancv\) is an ancilla variable that does not occur in \(\qterm\).
    \end{itemize}
\end{itemize}
The substitution for reset discards the content of $\qvar$ and replaces it by a fresh qumode in the vacuum: after renaming, $\csubs{\qterm}{\qvar}{\anc}$ no longer acts on $\qvar$, and the binder prepares $\anc$ in the vacuum.
These substitutions illustrate that programs with instructions that use measurement outcomes leave the finite polynomial algebra, thus introducing additional challenges.

\begin{lemma}[Dual map computation via substitutions]
\label{lemma:dual_map}
    For all admissible atomic programs \(\aprog \in \aprogs\), and quantum terms \(\qterm\), if $\sem{\qterm} \in \allhobsS$, then  \(\sem{\atomsubs{\aprog}{\qterm}} = \semD{\aprog}(\sem{\qterm})\).
\end{lemma}

\paragraph{Proof system.}
The proof system has two layers. First, backward term judgements compute the dual of a quantum program on a given quantum term. Then the proof rules in  Figure~\ref{fig:proof_system} substitute the resulting terms into the comparisons of a postcondition.
For all admissible \myPL~programs \(\prog \in \progs\) and all well-formed assertions \(\post \in \assertions\),
the weakest liberal precondition of $\prog$ with respect to $\post$ is a well-formed assertion \(\pre \in \assertions\) such that (i)~$\hltriple{\pre}{\prog}{\post}$ is valid, and (ii)~for all well-formed assertions $\pre'$, if $\hltriple{\pre'}{\prog}{\post}$ is valid, then $\pre' \Rightarrow \pre$.
Theorem~\ref{thm:proof_system} shows the proof system to be relatively complete for partial correctness according to this definition (see proof in Appendix Section~\ref{appendix:thm:proof_system}).

Given a program $\prog \in \progs$ and quantum terms $\qterm, \qterm' \in \qterms$,
we say that $\qterm'$ is an {\em exact dual} of $\prog$ with respect to $\qterm$ if $\sem{\qterm'} = \semD{\prog}(\sem{\qterm})$.
A {\em term judgement} $\tjudg{\prog}{\qterm'}{\qterm}$ states that $\qterm'$ is derivable as an exact dual of $\prog$ with respect to $\qterm$ using the rules of Figure~\ref{fig:term_judg}.
Lemma~\ref{lemma:judgements} shows soundness and completeness of these rules (see proof in Appendix Section~\ref{appendix:lemma:judgements}).

\begin{figure}[!t]
    \centering
    \scriptsize
    \[
    \begin{array}{ccc}
        \inferrule*[right=T-Atom]{
            \aprog \in \aprogs
            }{
                \tjudg{\aprog}{\atomsubs{\aprog}{\qterm}}{\qterm}
        } &
        \inferrule*[right=T-Seq]{
            \tjudg{\prog_1}{\qterm'}{\qterm} \quad \tjudg{\prog_0}{\qterm''}{\qterm'}
            }{
                \tjudg{\prog_0;\prog_1}{\qterm''}{\qterm}
        } &
        \inferrule*[right=T-Eq]{
            \tjudg{\prog}{\qterm_1'}{\qterm_1} \quad \sem{\qterm_1'} = \sem{\qterm'} \quad \sem{\qterm_1} = \sem{\qterm}
            }{
                \tjudg{\prog}{\qterm'}{\qterm}
        }
         \\[14pt]
        \multicolumn{3}{c}{
        \inferrule*[right=T-If]{
            \tjudg{\prog_0}{\qterm_0}{\qterm} \quad \tjudg{\prog_1}{\qterm_1}{\qterm}
            }{
                \tjudg{\IfElse{\bexpr}{\prog_0}{\prog_1}}{\guardterm{\bexpr}{\qterm_0} + \guardterm{\neg\bexpr}{\qterm_1}}{\qterm}
        }
        \qquad
        \inferrule*[right=T-While]{
            \tjudg{\prog}{\qterm}{\guardterm{\bexpr}{\qterm} + \guardterm{\neg\bexpr}{\qterm_0}}
            }{
                \tjudg{\While{\bexpr}{\prog}}{\guardterm{\bexpr}{\qterm} + \guardterm{\neg\bexpr}{\qterm_0}}{\qterm_0}
        }}
    \end{array}
    \]
    \Description{Backward term judgements}
    \caption{Rules for term judgements of the form $\tjudg{\prog}{\qterm'}{\qterm}$, which express that $\qterm'$ is derivable as an exact dual of $\prog$ with respect to $\qterm$.}
    \label{fig:term_judg}
\end{figure}

\begin{lemma}[Soundness and completeness of term judgements]
\label{lemma:judgements}
    For all admissible \myPL~programs $\prog \in \progs$ and quantum terms $\qterm, \qterm'$ with $\sem{\qterm} \in \allhobsS $, the following hold: (1)~if~$\tjudg{\prog}{\qterm'}{\qterm}$, then $\sem{\qterm'} = \semD{\prog}(\sem{\qterm})$, and (2)~if~$\sem{\qterm'} = \semD{\prog}(\sem{\qterm})$, then $\tjudg{\prog}{\qterm'}{\qterm}$.
\end{lemma}
\begin{figure}[!t]
    \centering
    \scriptsize
    \[
    \begin{array}{ccc}
        \inferrule*[right=Prog]{
            \forall \qterm_i \in \Terms{\post}.~\tjudg{\prog}{\qterm_i'}{\qterm_i}
            }{
                \hltriple{\csubs{\post}{\qterm_0,\ldots,\qterm_{|\Terms{\post}|-1}}{\qterm_0', \ldots, \qterm_{|\Terms{\post}|-1}'}}{\prog}{\post}
            }
        &
        \inferrule*[right=Conseq]{
            \pre \Rightarrow \pre'  \quad \hltriple{\pre'}{\prog}{\post'} \quad \post' \Rightarrow \post
            }{
            \hltriple{\pre}{\prog}{\post}
        } \\[14pt]
    \end{array}
    \]
    \Description{Proof system}
    \caption{Proof system for $\myHL$, where $\Terms{\post}$ is the set of quantum terms occurring in $\post$.}
    \label{fig:proof_system}
\end{figure}

\begin{theorem}[Proof system] Let $\vdash$ denote syntactic derivability using the proof system of Figure~\ref{fig:proof_system}. For all admissible \myPL~quantum programs $\prog \in \progs$ and well-formed assertions $\pre, \post \in \assertions$, the following hold: (1)~if $\vdash \hltriple{\pre}{\prog}{\post}$, then $\hltriple{\pre}{\prog}{\post}$ is valid; 
    (2)~the rule \textsc{Prog} derives the weakest liberal precondition of $\prog$ with respect to $\post$; and (3)~if $\hltriple{\pre}{\prog}{\post}$ is valid, then $\vdash\hltriple{\pre}{\prog}{\post}$ relative to oracles deciding the entailment $\Rightarrow$ of assertions.
\label{thm:proof_system}
\end{theorem}

%% file: sections/_7_experiments.tex
\newcommand{\codeComment}[1]{\textcolor{commentColor}{\footnotesize \# #1}}
\newcommand{\asgn}{\mathbin{:=}}
\newcommand{\lnum}[1]{\textcolor{hoareNum}{\footnotesize #1}}   % text-mode now
\newcommand{\alghead}[2]{\textbf{\textsf{#1}}\,($#2$)} % name + params
\section{Experiments}
\label{sec:experiments}
We implemented a prototype verifier in Python 3.10.13 that derives weakest preconditions for continuous-variable quantum programs. The implementation uses SymPy 1.14.0 to manipulate polynomials over quadrature operators symbolically.
We evaluate the verifier in three representative use cases.
First, we verify standard textbook algorithms\footnote{Some programs do not use instructions that are native in our programming language, but their decomposition can be found in Table~\ref{tab:main_decomp}.} from the literature against their specifications and show that weakest preconditions recover known correctness results while also deriving symbolic formulas that characterize the inherent noise of practical implementations.
Second, we prove the equivalence between unitary gates and their decompositions used in Strawberry Fields. Third, we use weakest preconditions to estimate the Fock-space truncation required for classical simulation.

\subsection{Program verification}
\subsubsection{Homodyne measurements via Fock measurements~\cite{Lloyd99,Braunstein05,Leonhardt95}}
\label{sec:exp_fock}
Homodyne measurements describe a projective measurement onto the unphysical eigenstates of the position operator (a prior rotation allows momentum to be measured instead).

The program of Figure~\ref{fig:homodyne_sc}(a) measures the position of $\qvar_0$: an ancilla $\qvar_1$,
displaced by a symbolic constant $\real \in \Reals$, is coupled to the target $\qvar_0$, and the symbolic position $x \in \Reals$ of $\qvar_0$ is known to be proportional to the difference of the two photon-number
outcomes: i.e.,
$
\expectedV{\X_{\qvar_0}}
=
    (\expectedV{\nop_{\qvar_1}} - \expectedV{\nop_{\qvar_0}})
/
    \real
 = x. 
$
Existing descriptions of this measurement scheme recommend choosing a sufficiently large displacement amplitude $\real$, but do not quantify how large it should be or how it affects the measurement accuracy~\cite{Lloyd99,Braunstein05,Leonhardt95}. 
The weakest precondition derived in Figure~\ref{fig:homodyne_sc}(a) formally proves that the photon-number difference is an unbiased estimator of the position expectation: the expected measurement outcome is exactly $\expectedV{\X_{\qvar_0}}$, independently of the displacement value $\real$.
More importantly, our derivation quantifies the effect of a finite displacement on the measurement variance. While the displacement value shows no relevance for the expected mean value, 
the weakest precondition for the quadratic observable
 $(\nop_{\qvar_1} - \nop_{\qvar_0})^2/{r^2}$ reveals an additional single-shot noise term 
 ${\nop_{\qvar_0}/r^{2}}$ which is added to the true position variance 
 $\expectedV{\X^2_{\qvar_0}} - \expectedV{\X_{\qvar_0}}^2$.
Since ${\nop_{\qvar_0}/r^{2}}$ is a positive semidefinite observable, the single-shot noise is always nonnegative (vanishing only when $\qvar_0$ is in the vacuum state), and decreases as $|r|$ increases.
Consequently, the derived weakest precondition provides a quantitative criterion for selecting the displacement amplitude in practice, explicitly characterizing the trade-off between finite resources and measurement accuracy.
\begin{figure}[!tbh]
    \centering
\input{img/verification/t_HomodyneMeas.tex}
    \hfill    \input{img/verification/t_SuperdenseCoding}
    \caption{(a)~Hoare triple for the program $\progName{HomodyneMeas}$ that measures the position of $\qvar_0$ via Fock measurements using an ancilla $\qvar_1$. (b)~Hoare triple for the 
        $\progName{SuperdenseCoding}$ program, with squeezing parameter $r$, and 
        two real numbers (x, p) that correspond to the message that Alice ($\qvar_0$) wants to send to Bob ($\qvar_1$).}
\label{fig:homodyne_sc}
\end{figure}

\subsubsection{Deutsch-Jozsa algorithm~\cite{Jozsa92, Braunstein03}}
The Deutsch-Jozsa algorithm decides whether a given oracle function is
constant or balanced. 
In the continuous-variable setting the oracle is an additive oracle
$
  U_f:\ \ket{x}\ket{y}\ \mapsto\ \ket{x}\,\ket{y + f(x)}
$, where $x,y\in\Reals$,
acting on a query mode $\qvar_0$ and an ancilla mode $\qvar_1$, and  $f:\Reals\to\Reals$ is a function over reals.
We consider the linear oracle
$f(x)=s\cdot x$ that is constant if $s=0$ and balanced otherwise.
For this linear oracle the 
derivation shown in Figure~\ref{fig:dj_tel}(a) yields an expected position of  $x = -(\pi\cdot s)/2$ for $\qvar_0$, which is zero precisely when the oracle is constant. 
Thus, the derived expectation reproduces the correctness criterion of the Deutsch-Jozsa algorithm.
The derived expectation immediately yields the slope of the linear oracle,
$
s=-2\cdot\expectedV{\X_{\qvar_0}}/\pi,
$
which corresponds to the continuous-variable analogue of the Bernstein-Vazirani algorithm.
The position variance reveals that finite squeezing introduces a noise term proportional to $(s^2+1)\cdot e^{-2r}$, that vanishes exponentially as $r\rightarrow \infty$.

\begin{figure}[tbh]
    \centering
    \input{img/verification/t_DeutschJozsa}
    \hfill
    \input{img/verification/t_StateTeleportation}
    \caption{(a)~Hoare triple for the continuous-variable Deutsch-Jozsa program
$\progName{DeutschJozsa}$, with squeezing parameter $r$, oracle slope $s$,
query mode $\qvar_0$, and ancilla $\qvar_1$. 
 (b)~Hoare triple for the state teleportation program
$\progName{StateTeleportation}$.}
\label{fig:dj_tel}
\end{figure}

\subsubsection{Superdense coding~\cite{Ban99,Braunstein00,Braunstein05}}

Superdense coding enables Alice to communicate two real numbers $(x,p)$ to Bob by transmitting a single qumode, provided they initially share an entangled resource. Figure~\ref{fig:homodyne_sc}(b)  implements the continuous-variable protocol. Lines~2-6 prepare a two-mode squeezed state with squeezing parameter $r$, after which Alice receives $\qvar_0$, Bob receives $\qvar_1$, and Alice encodes her message by applying a displacement to $\qvar_0$ (line~9). Alice then sends her qumode to Bob, who decodes the message using a beam splitter (line~12) followed by homodyne measurements.

The protocol relies on a shared entangled resource that approaches an ideal Einstein-Podolsky-Rosen (EPR) state only in the limit $r\rightarrow\infty$. Consequently, perfect communication is achieved only via infinite-squeezing. Our derivation proves that Bob recovers the encoded message by measuring the position quadrature of $\qvar_0$ and the momentum quadrature of $\qvar_1$, yielding the expected values
$\expectedV{\X_{\qvar_0}} = x$ and $
\expectedV{\POp_{\qvar_1}} = p$.
The weakest precondition derived for the quadratic terms shows that finite squeezing introduces error terms
$
    \expectedV{\X_{\qvar_0}^2} - \expectedV{\X_{\qvar_0}}^2 = \expectedV{\POp_{\qvar_1}^2} - \expectedV{\POp_{\qvar_1}}^2 = e^{- 2 r}/2
$ in the variance. 
Thus, our verifier not only establishes the correctness of the protocol but also quantifies the degradation caused by finite squeezing, providing an explicit relationship between the squeezing parameter $\real$ and the communication accuracy.

\subsubsection{Quantum state teleportation~\cite{Bennett93,Furusawa11,Braunstein98}}
Continuous-variable quantum teleportation allows a sender (Alice) to transmit the quantum state of a qumode to a receiver (Bob) using a shared entangled resource and classical communication. Ideally, Bob's final state is identical to Alice's input state. Figure~\ref{fig:dj_tel}(b) shows an equivalent purely unitary realization of this protocol.
The program first prepares a two-mode squeezed (EPR-type) state shared between Alice ($\qvar_1$) and Bob ($\qvar_2$) using the squeezing parameter $r$ (lines~2-4). Alice then couples the input mode $\qvar_0$ to her half of the entangled resource and transfers the measurement outcomes to Bob through controlled operations (lines~6-10). As in the previous examples, the shared resource approaches an ideal EPR state only in the limit $r\rightarrow\infty$, and therefore perfect teleportation is achieved only in the infinite-squeezing limit.
The derived triple quantifies the effect of finite squeezing. The weakest precondition shows that Bob's quadrature operators 
reveal additive noise terms proportional to $e^{-r}$ in the first moments of his final quantum state. 
Consequently, Bob's expected position and momentum converge to those of Alice's input state as $r\rightarrow\infty$, while the teleportation error decreases exponentially with the squeezing parameter.
In the Appendix Section~\ref{appendix:state_tel} we explain how 
the derived weakest preconditions can be used to show that Bob's final state coincides with Alice's input state in the infinite-squeezing limit.

\subsection{Program equivalence}
In this subsection, we use our proof system to prove the equivalence of a unitary gate and its decomposition provided by Strawberry Fields~\cite{Killoran19}.
The documentation of Strawberry Fields provides both the unitary gate decomposition and the transformation of canonical observables under the dual map of each unitary. By Proposition~\ref{prop:unitary_eq}, the weakest preconditions that our tool  derives can be used to determine whether two unitaries are equivalent up to a global phase.

Table~\ref{tab:main_decomp} shows the target unitaries, their decompositions using instructions of our programming language, along with the corresponding preconditions and postconditions. The decompositions correspond to single- and two-qumode gates. Accordingly, we verify two and four triples respectively.
Some of these gates are parametric, i.e., they rely on real-valued parameters, such as the $\Ins{X}(\qvar_0, x)$, whose parameter $x \in \Reals$ corresponds to the amount by which a qumode is displaced in the position basis. 
Some decompositions use parameters (e.g. $\phi$ or $\theta$) that depend on the input parameters of the target unitary.
Names in the first column of Table~\ref{tab:main_decomp} and parameters follow the same convention as the Strawberry Fields documentation~\cite{Killoran19}.

\input{img/decompositions/summary_decomps.tex}

\subsection{Resource estimation for classical simulation}
Classical simulation of continuous-variable quantum programs requires truncating the infinite-dimensional Fock space. Given a truncation dimension $\hsDim$, only the basis states
$
\{\ket{0},\ket{1},\ldots,\ket{\hsDim-1}\}
$
are retained, while the remaining probability mass
$\Pr(n\geq\hsDim)$ is discarded. A natural question is therefore how large $\hsDim$ must be to guarantee a prescribed simulation accuracy.
Our verifier answers this question by deriving weakest preconditions for the number operator. For $\hsDim>\expectedV{\nop}$, the one-sided Chebyshev-Cantelli inequality~\cite{Boucheron13} gives the equation shown in Figure~\ref{tab:number_op}.
Consequently, determining a truncation dimension that satisfies a desired error tolerance requires only the expectation and variance of the number operator, that is, the expectations of $\nop$ and $\nop^2$.

The table in Figure~\ref{tab:number_op} lists the weakest precondition our tool derives for
$\nop$ under each single-qumode instruction applied to a quantum variable
$\qvar$, with postcondition $\{\nop=\nat\}$. 
We do not tabulate $\nop^2$ separately. For every instruction in our set the weakest precondition of $\nop^2$ is the \emph{square} of the weakest precondition of $\nop$: conjugation is multiplicative under unitary evolution and the other atomic instructions are diagonal in the Fock basis.

The derived weakest preconditions provide immediate insight into the classical simulation cost of each instruction.
Measurements and rotations leave the number operator unchanged and therefore preserve both its expectation and variance, so they do not increase the required truncation dimension.
The weakest precondition derived for the displacement instruction shows that linear terms may partially compensate for the quadratic offset, so depending on the initial state, a displacement may increase or decrease the required truncation.
The cubic-phase gate introduces the additional terms
whose dominant contribution $r^2\cdot \expectedV{\X^4}/2$ is always nonnegative. 
Consequently, cubic phase operations generally require substantially larger truncation dimensions.

\paragraph{Example.}
Consider the vacuum state $\ket{0}\bra{0}$ and the program $\prog = \kw{D}(\qvar,x,0)$  for some $x\in \Reals$. 
From Figure~\ref{tab:number_op}, the derived weakest precondition  is $\nop+\sqrt{2}\cdot x\cdot\X+x^2 = n$.
Evaluating the weakest precondition on the vacuum yields 
$\expectedV{\nop} = x^2$ and $\expectedV{\nop^2}   = x^4 + x^2$.
Substituting these moments into the Chebyshev-Cantelli inequality  gives the smallest truncation dimension satisfying a desired error tolerance.
For example, if $x=2$ and the discarded probability mass must not exceed $10\%$, then the bound yields $\hsDim=10$. Therefore, simulating $\kw{D}(\qvar,2,0)$ from $\ket{0}\bra{0}$ using the first ten Fock basis states guarantees that at least $90\%$ of the probability mass is retained.

\begin{figure}[tbh]
    \footnotesize
    \centering
    % Left side: The Formula
    \begin{minipage}[c]{0.6\textwidth}
        \centering
        \input{img/number_op.tex}
    \end{minipage}
    \hfill % Pushes the two blocks apart
    % Right side: The Table
    \begin{minipage}[c]{0.38\textwidth}
        \centering
        Chebyshev-Cantelli inequality 
        \begingroup
        \small
        \begin{equation*}
            \Pr(n\geq\hsDim) \leq
  \frac{\expectedV{\nop^2}-\expectedV{\nop}^2}
       {\expectedV{\nop^2}-\expectedV{\nop}^2+(\hsDim-\expectedV{\nop})^2}
    %    \label{eq:cantelli}
        \end{equation*}
        \endgroup
    \end{minipage}
    \Description{Number operator table.}
    \caption{On the left, we have the weakest precondition derived to determine how the expectation value of the number operator $\nop$ varies on all of our single-qumode instructions.
    The weakest preconditions, shown in the first column, are derived for the postcondition $\{\nop ~ = \nat \}$, where $\nat$ is a symbolic real variable. The Chebyshev-Cantelli inequality, on the right, can be used to bound the accuracy of classical simulations.}
    \label{tab:number_op}
\end{figure}

% verification of correct preparation of the GBS resource state,

%% file: img/verification/t_HomodyneMeas.tex
\begin{subfigure}[b]{0.4\textwidth}
\footnotesize
    \centering
    \begin{tcolorbox}[
        hbox,                       % <-- width fits the content
        colback=hoareBg,
        colframe=hoareFrame,
        boxrule=0.6pt, arc=3pt,
        boxsep=0pt,
        left=9pt, right=9pt, top=7pt, bottom=7pt,
    ]
    \renewcommand{\arraystretch}{1.15}%
    \arrayrulecolor{hoareFrame}%
        \begin{tabular}{@{}r@{\;\;}l@{}}
            \multicolumn{2}{@{}l@{}}{\textbf{\textsf{HomodyneMeas}}\,($r$)} \\   % signature region
            \noalign{\vskip3pt}\hline\noalign{\vskip3pt}
            \multicolumn{2}{@{}l@{}}{$
            \assert{\, \left(\X_{\qvar_0} = x\right) \land \left(\X_{\qvar_0}^2 + \nop_{\qvar_0}/r^2 = v\right)  \,}$}\\                         % precondition
            \noalign{\vskip3pt}\hline\noalign{\vskip3pt}
            
            % program instructions
			\lnum{1} & $\ResetI(\qvar_1); $ \\
			\lnum{2} & \\
			\lnum{3} & $\DispI{\qvar_1}{\real/\sqrt{2}}{0}; $ \codeComment{displace ancilla by $r$}\\
			\lnum{4} & $\BSI{\qvar_0}{\qvar_1}{\pi/4}{0}; $ \\
			\lnum{5} & \\
			\lnum{6} & \codeComment{Fock measurements}\\
			\lnum{7} & $\MeasI{\cvar_0}{\qvar_0};~\MeasI{\cvar_1}{\qvar_1}; $ \\
            
            % end program instructions
            \noalign{\vskip3pt}\hline\noalign{\vskip3pt}
            \multicolumn{2}{@{}l@{}}{$\assert{\,\left(\frac{\nop_{\qvar_1} - \nop_{\qvar_0}}{r} = x\right) \land \left(\frac{\left(\nop_{\qvar_1} - \nop_{\qvar_0}\right)^2}{r^2} = v\right) \,}$}\\                         % postcondition
        \end{tabular}
    \end{tcolorbox}
\end{subfigure}

%% file: img/verification/t_SuperdenseCoding.tex
\begin{subfigure}[b]{0.55\textwidth}
    \centering
    \footnotesize
    \begin{tcolorbox}[
        hbox,                       % <-- width fits the content
        colback=hoareBg,
        colframe=hoareFrame,
        boxrule=0.6pt, arc=3pt,
        boxsep=0pt,
        left=9pt, right=9pt, top=7pt, bottom=7pt,
    ]
    \renewcommand{\arraystretch}{1.15}%
    \arrayrulecolor{hoareFrame}%
        \begin{tabular}{@{}r@{\;\;}l@{}}
            \multicolumn{2}{@{}l@{}}{\textbf{\textsf{SuperdenseCoding}}\,($r,\,x,\,p$)} \\   % signature region
            \noalign{\vskip3pt}\hline\noalign{\vskip3pt}
            \multicolumn{2}{@{}l@{}}{$\assert{\,\left( x = x\right) \land \left(p = p\right) \land \left(x^{2} + \frac{e^{- 2 r}}{2} = x^2\right) \land \left(p^{2} + \frac{e^{- 2 r}}{2} = p^2\right) \,}$}\\
            \noalign{\vskip3pt}\hline\noalign{\vskip3pt}
            
            % program instructions
			\lnum{1} & \codeComment{create entangled state}\\
			\lnum{2} & $\ResetI(\qvar_0); $ \\
			\lnum{3} & $\ResetI(\qvar_1); $ \\
			\lnum{4} & $\SqueezeI{\qvar_0}{-\real}; $ \\
			\lnum{5} & $\SqueezeI{\qvar_1}{\real}; $ \\
			\lnum{6} & $\BSI{\qvar_0}{\qvar_1}{\pi/4}{0}; $ \\
			\lnum{7} & \\
			\lnum{8} & \codeComment{encode Alice message (x, p) and send $\qvar_0$ to Bob}\\
			\lnum{9} & $\DispI{\qvar_0}{x}{p}; $ \\
			\lnum{10} & \\
			\lnum{11} & \codeComment{Bob decodes (he now holds both modes)}\\
			\lnum{12} & $\BSI{\qvar_0}{\qvar_1}{\pi/4}{0}; $ \\
			\lnum{13} & \codeComment{$\expectedV{\X_{\qvar_0}} = x$ and  $\expectedV{\POp_{\qvar_1}} = p$}\\
            % end program instructions
            \noalign{\vskip3pt}\hline\noalign{\vskip3pt}
            \multicolumn{2}{@{}l@{}}{$\assert{\,\left(\X_{\qvar_0} = x\right) \land \left(\POp_{\qvar_1} = p\right) \land \left(\X_{\qvar_0}^2 = x^2\right) \land \left(\POp_{\qvar_1}^2 = p^2\right)\,}$}\\                         % postcondition
        \end{tabular}
    \end{tcolorbox}
    \Description{ Hoare triple for the 
        $\textsf{SuperdenseCoding}(r, x, p)$ program, with squeezing parameter $r$, and 
        two real numbers (x, p) that correspond to the message that Alice ($\qvar_0$) wants to send to Bob ($\qvar_1$). }
\label{fig:SuperdenseCoding}
\end{subfigure}

%% file: img/verification/t_DeutschJozsa.tex
\begin{subfigure}[b]{0.48\textwidth}
\footnotesize
    \centering
    \begin{tcolorbox}[
        hbox,                       % <-- width fits the content
        colback=hoareBg,
        colframe=hoareFrame,
        boxrule=0.6pt, arc=3pt,
        boxsep=0pt,
        left=9pt, right=9pt, top=7pt, bottom=7pt,
    ]
    \renewcommand{\arraystretch}{1.15}%
    \arrayrulecolor{hoareFrame}%
        \begin{tabular}{@{}r@{\;\;}l@{}}
            \multicolumn{2}{@{}l@{}}{\textbf{\textsf{DeutschJozsa}}\,($r,\,s$)} \\   % signature region
            \noalign{\vskip3pt}\hline\noalign{\vskip3pt}
            \multicolumn{2}{@{}l@{}}{$\assert{\, \left(-\frac{\pi s}{2} = x\right) \land \left(\frac{(\pi s)^2}{4} + \frac{(s^2+1)}{2e^{2r}} = v\right) \,}$}\\                         % precondition
            \noalign{\vskip3pt}\hline\noalign{\vskip3pt}
            
            % program instructions
			\lnum{1} & \codeComment{prepare initial states}\\
			\lnum{2} & $\ResetI({\qvar_0}); $ $\ResetI({\qvar_1}); $ \\
			\lnum{3} & $\SqueezeI{\qvar_0}{\real}; $ $\kw{S}({\qvar_1}, r); $ \\
			\lnum{4} & $\DispI{\qvar_1}{\sqrt{2} \pi/4}{0}; $ \\
			\lnum{5} & $\RotI{\qvar_0}{\pi/2}; $ $\RotI{\qvar_1}{\pi/2}; $ \\
			\lnum{6} & \\
			\lnum{7} & \codeComment{linear oracle function}\\
            \lnum{8} & \codeComment{(constant if $s=0$, otherwise balanced)}\\
			\lnum{9} & $\CXI{\qvar_0}{\qvar_1}{s}; $ \\
			\lnum{10} & \\
			\lnum{11} & $\RotI{\qvar_0}{-\pi/2}; $ \\
			\lnum{12} & \codeComment{$\expectedV{\X_{\qvar_0}} = 0$ if $s=0$, otherwise $\expectedV{\X_{\qvar_0}} \neq 0$}\\

            % end program instructions
            \noalign{\vskip3pt}\hline\noalign{\vskip3pt}
            \multicolumn{2}{@{}l@{}}{$\assert{\,\left(\X_{\qvar_0} = x \right) \land \left(\X_{\qvar_0}^2 = v\right) \,}$}\\                         % postcondition
        \end{tabular}
    \end{tcolorbox}
    \Description{Hoare triple for the continuous-variable Deutsch-Jozsa program.}
    \label{fig:DeutschJozsa}
\end{subfigure}

%% file: img/verification/t_StateTeleportation.tex
\begin{subfigure}[b]{0.48\textwidth}
    \centering
    \small
    \begin{tcolorbox}[
        hbox,                       % <-- width fits the content
        colback=hoareBg,
        colframe=hoareFrame,
        boxrule=0.6pt, arc=3pt,
        boxsep=0pt,
        left=9pt, right=9pt, top=7pt, bottom=7pt,
    ]
    \renewcommand{\arraystretch}{1.15}%
    \arrayrulecolor{hoareFrame}%
        \begin{tabular}{@{}r@{\;\;}l@{}}
            \multicolumn{2}{@{}l@{}}{\textbf{\textsf{StateTeleportation}}\,($r$)} \\   % signature region
            \noalign{\vskip3pt}\hline\noalign{\vskip3pt}
            \multicolumn{2}{@{}l@{}}{$\assert{\,\left( \X_{\qvar_0} + \frac{\sqrt{2}\cdot\X_{\qvar_2}}{e^{r}} = x\right) \land \left(\POp_{\qvar_0} + \frac{\sqrt{2}\cdot\POp{\qvar_1}}{e^{r} } = p\right) \,}$}\\                         % precondition
            \noalign{\vskip3pt}\hline\noalign{\vskip3pt}
            
            % program instructions
			\lnum{1} & \codeComment{create entangled resource}\\
			\lnum{2} & $\SqueezeI{\qvar_1}{ -\real}; $ \codeComment{squeeze Alice's momentum}\\
			\lnum{3} & $\SqueezeI{\qvar_2}{\real}; $ \codeComment{squeeze Bob's position}\\
			\lnum{4} & $\BSI{\qvar_1}{\qvar_2}{\pi/4}{0}; $ \\
			\lnum{5} & \\
			\lnum{6} & $\BSI{\qvar_0}{\qvar_1}{\pi/4}{0};$ \codeComment{Alice mixes the input}\\
			\lnum{7} & $
            \CXI{\qvar_0}{\qvar_2}{\sqrt{2}}; $ \codeComment{transmit  position}\\
			\lnum{8} & $\RotI{\qvar_1}{-\pi/2}; $ \codeComment{change of basis}\\
			\lnum{9} & $\CZI{\qvar_1}{\qvar_2}{\sqrt{2}}; $ \codeComment{transmit momentum}\\
			\lnum{10} & $\RotI{\qvar_1}{\pi/2};$ \codeComment{restore basis}\\

            % end program instructions
            \noalign{\vskip3pt}\hline\noalign{\vskip3pt}
            \multicolumn{2}{@{}l@{}}{$\assert{\,\left(\X_{\qvar_2} = x\right) \land \left(\POp_{\qvar_2} = p\right) \,}$}\\                         % postcondition
        \end{tabular}
    \end{tcolorbox}
    \Description{Hoare triple for the state teleportation program.}
\label{fig:StateTeleportation}
\end{subfigure}

%% file: img/decompositions/summary_decomps.tex
\begin{table}[!t]
\footnotesize
\centering
	\begin{tabular}{r|c|l|c|}
		 target unitary & precondition & \multicolumn{1}{c|}{decomposition} & postcondition\\
		\hline
		\hline
		 \multirow{2}{*}{$\Ins{X}(\qvar_0, x)$} & $\assert{\, x + \X_{\qvar_0} = x \,}$ & \multirow{2}{*}{\makecell[l]{$\kw{D}(\qvar_0, \sqrt{2} x/2, 0); $}} & $\assert{\, \X_{\qvar_0} = x \,}$  \\
		\cline{2-2}\cline{4-4}
		 & $\assert{\, \POp_{\qvar_0} = p \,}$ & & $\assert{\, \POp_{\qvar_0} = p \,}$ \\
		\hline
		 \multirow{2}{*}{$\Ins{Z}(\qvar_0, p)$} & $\assert{\, \X_{\qvar_0} = x \,}$ & \multirow{2}{*}{\makecell[l]{$\kw{D}(\qvar_0, 0, \sqrt{2} p/2); $}} & $\assert{\, \X_{\qvar_0} = x \,}$  \\
		\cline{2-2}\cline{4-4}
		 & $\assert{\, p + \POp_{\qvar_0} = p \,}$ & & $\assert{\, \POp_{\qvar_0} = p \,}$ \\
		\hline
		 \multirow{2}{*}{$\Ins{Fourier}(\qvar_0)$} & $\assert{\, - \POp_{\qvar_0} = p \,}$ & \multirow{2}{*}{\makecell[l]{$\kw{R}(\qvar_0, \pi/2); $}} & $\assert{\, \X_{\qvar_0} = x \,}$  \\
		\cline{2-2}\cline{4-4}
		 & $\assert{\, \X_{\qvar_0} = x \,}$ & & $\assert{\, \POp_{\qvar_0} = p \,}$ \\
		\hline
		 \multirow{2}{*}{$\Ins{P}(\qvar_0, s)$} & 
         $\assert{\, \X_{\qvar_0} = x \,}$ & 
         \multirow{2}{*}{\makecell[l]{$\kw{R}(\qvar_0, - \phi/2);~\kw{S}(\qvar_0, r); $\\ $\kw{R}(\qvar_0, \phi/2);~\kw{R}(\qvar_0, \theta); $}} & $\assert{\, \X_{\qvar_0} = x \,}$  \\
		\cline{2-2}\cline{4-4}
		 & $\assert{\, s \X_{\qvar_0} + \POp_{\qvar_0} = p \,}$ & & $\assert{\, \POp_{\qvar_0} = p \,}$ \\
		\hline
		 \multirow{4}{*}{$\CXI{\qvar_0}{\qvar_1}{s}$} & $\assert{\, \X_{\qvar_0} = x_0 \,}$ & \multirow{4}{*}{\makecell[l]{$\kw{BS}(\qvar_0, \qvar_1, \theta, 0); $\\ $\kw{S}(\qvar_0, r);~\kw{S}(\qvar_1, - r); $\\ $\kw{BS}(\qvar_0, \qvar_1, \theta + \pi/2, 0); $}} & $\assert{\, \X_{\qvar_0} = x_0 \,}$  \\
		\cline{2-2}\cline{4-4}
		 & $\assert{\, - s \POp_{\qvar_1} + \POp_{\qvar_0} = p_0 \,}$ & & $\assert{\, \POp_{\qvar_0} = p_0 \,}$ \\
		\cline{2-2}\cline{4-4}
		 & $\assert{\, s \X_{\qvar_0} + \X_{\qvar_1} = x_1 \,}$ & & $\assert{\, \X_{\qvar_1} = x_1 \,}$ \\
		\cline{2-2}\cline{4-4}
		 & $\assert{\, \POp_{\qvar_1} = p_1 \,}$ & & $\assert{\, \POp_{\qvar_1} = p_1 \,}$ \\
		\hline
		 \multirow{4}{*}{$\CZI{\qvar_0}{\qvar_1}{s}$} & $\assert{\, \X_{\qvar_0} = x_0 \,}$ & \multirow{4}{*}{\makecell[l]{
            $\kw{R}(\qvar_1, - \pi/2);~ 
            \kw{BS}(\qvar_0, \qvar_1, \theta, 0); $\\ $\kw{S}(\qvar_0, r);~\kw{S}(\qvar_1, - r); $\\ $\kw{BS}(\qvar_0, \qvar_1, \theta + \pi/2, 0); $\\ $\kw{R}(\qvar_1, \pi/2); $}} & $\assert{\, \X_{\qvar_0} = x_0 \,}$  \\
		\cline{2-2}\cline{4-4}
		 & $\assert{\, s \X_{\qvar_1} + \POp_{\qvar_0} = p_0 \,}$ & & $\assert{\, \POp_{\qvar_0} = p_0 \,}$ \\
		\cline{2-2}\cline{4-4}
		 & $\assert{\, \X_{\qvar_1} = x_1 \,}$ & & $\assert{\, \X_{\qvar_1} = x_1 \,}$ \\
		\cline{2-2}\cline{4-4}
		 & $\assert{\, s \X_{\qvar_0} + \POp_{\qvar_1} = p_1 \,}$ & & $\assert{\, \POp_{\qvar_1} = p_1 \,}$ \\
		\hline
		 \multirow{4}{*}{$\Ins{MZ}(\qvar_0, \qvar_1, \phi_{in}, \phi_{ex})$} & $\assert{\, \ldots \,}$ & \multirow{4}{*}{\makecell[l]{$\kw{R}(\qvar_0, \phi_{ex}); $\\ $\kw{BS}(\qvar_0, \qvar_1, \pi/4, \pi/2); $\\ $\kw{R}(\qvar_0, \phi_{in}); $\\ $\kw{BS}(\qvar_0, \qvar_1, \pi/4, \pi/2); $}} & $\assert{\, \X_{\qvar_0} = x_0 \,}$  \\
		\cline{2-2}\cline{4-4}
		 & $\assert{\, \ldots \,}$ & & $\assert{\, \POp_{\qvar_0} = p_0 \,}$ \\
		\cline{2-2}\cline{4-4}
		 & $\assert{\, \ldots \,}$ & & $\assert{\, \X_{\qvar_1} = x_1 \,}$ \\
		\cline{2-2}\cline{4-4}
		 & $\assert{\, \ldots \,}$ & & $\assert{\, \POp_{\qvar_1} = p_1 \,}$ \\
		\hline
		 \multirow{4}{*}{$\Ins{S2}(\qvar_0, \qvar_1, r, \phi)$} & $\assert{\, \ldots \,}$ & \multirow{4}{*}{\makecell[l]{$\kw{BS}(\qvar_0, \qvar_1, \pi/4, 0); 
         ~\kw{R}(\qvar_0, - \phi/2);$\\
         $\kw{S}(\qvar_0, r); ~\kw{R}(\qvar_0, \phi/2);$ \\ 
         $\kw{R}(\qvar_1, - \phi/2); ~\kw{S}(\qvar_1, - r);$\\ $\kw{R}(\qvar_1, \phi/2);~\kw{BS}(\qvar_0, \qvar_1, - \pi/4, 0); $}} & $\assert{\, \X_{\qvar_0} = x \,}$  \\
		\cline{2-2}\cline{4-4}
		 & $\assert{\, \ldots \,}$ & & $\assert{\, \POp_{\qvar_0} = p \,}$ \\
		\cline{2-2}\cline{4-4}
		 & $\assert{\, \ldots \,}$ & & $\assert{\, \X_{\qvar_1} = x_1 \,}$ \\
		\cline{2-2}\cline{4-4}
		 & $\assert{\, \ldots \,}$ & & $\assert{\, \POp_{\qvar_1} = p_1 \,}$ \\
		\hline
	\end{tabular}
	\caption{Verified gate decompositions.}
	\label{tab:main_decomp}
\end{table}

%% file: img/number_op.tex
\begin{table}[H]
\footnotesize
\centering
	\begin{tabular}{c|c}
		 $\weakp{\prog}{\{~ \nop = \nat ~\}}$ & $\prog$ \\
		\hline
		\hline
		$\nop = \nat$ & $\kw{Meas}(\qvar); $ \\
		\hline
		$\nop + p^{2} + \sqrt{2}\cdot p\cdot\POp + x^{2} + \sqrt{2}\cdot x \cdot\X = \nat$ & $\kw{D}(\qvar, x, p); $ \\
		\hline
		$\nop + (e^{2 r} \cdot\POp^{2} - \POp^{2} - \X^{2} + e^{- 2 r} \cdot\X^{2})/2 = \nat$ & $\kw{S}(\qvar, r); $ \\
		\hline
		$\nop = \nat$ & $\kw{R}(\qvar, r); $ \\
		\hline
		$\nop + r \cdot\left(r \cdot\X^{4} + \POp\cdot \X^{2} + \X^{2} \cdot\POp\right)/2 = \nat$ & $\kw{V}(\qvar, r); $ \\
		\hline
        $0 = \nat$ & $\ResetI(\qvar); $ \\
		\hline
	\end{tabular}
\end{table}

%% file: sections/conclusion.tex
\section{Conclusion}
\label{sec:conclusion}
Despite being a solid candidate for quantum advantage, 
CQC lacks formal verification techniques.
Since DQC techniques cannot be directly extended, 
CQC requires separate logical foundations.
Our work fills this need by introducing the first Hoare logic 
for CQC, together with a relative-completeness proof for a universal CQC programming language and an automatic calculator of weakest preconditions. 
Our case studies show that the framework is useful beyond verifying functional correctness.
The derived weakest preconditions compute symbolic constraints for physically relevant quantities---such as the finite-squeezing noise---which can guide parameter selection in practical realizations;
they can be used to determine program equivalence for quantum-gate decomposition; 
and they provide rigorous resource bounds for the classical simulation of CQC programs.
Our results demonstrate that Hoare-style reasoning remains possible for continuous-variable quantum computing despite the challenges imposed by infinite-dimensional Hilbert spaces and unbounded measurement outcomes; 
they provide a logical foundation for future continuous-variable quantum programming languages, compilers, and program-analysis methods.

%% file: sections/appendix.tex
\newtheorem*{lemma*}{Lemma}
\newtheorem*{theorem*}{Theorem}
\newtheorem*{proposition*}{Proposition}

\section{Appendix}

\subsection{Supplementary proofs}
This section contains supplementary proofs that are used by theorems, lemmas, and  propositions in the main text.

\subsubsection{Appendix: Powers of ladder operators }
\label{appendix:prop:powers_ladder}
\begin{proposition}[Powers of ladder operators] For all nonnegative integers $m \in \Nats$, we have 
    \begin{equation}
        \dop^m = \sum_{i=0}^{\infty} \left(\prod_{j = 1}^m\sqrt{(i+j)}\right) \ket{i}\bra{i+m}, \quad\quad \text{ and } \quad\quad (\cop)^m = \sum_{i=m}^{\infty} \left(\prod_{j = 0}^{m-1}\sqrt{i-j}\right) \ket{i}\bra{i-m}.
        \label{eq:powers_ladder}
    \end{equation}
    \label{prop:powers_ladder}
\end{proposition}
\begin{proof} The proof follows by induction.
\begin{itemize}
    \item For the annihilation operator ($\dop$).
    \begin{itemize}
        \item \textit{Base case: $m=0$.} Gives $\sum_{i=0}^\infty \ket{i}\bra{i}$.
        \item \textit{Inductive case.} 
        Since the Fock basis forms an orthonormal basis the formula for $m+1$ gives
        \begin{equation*}
            \begin{aligned}
                \dop^{m+1} &=  \dop^{m}\dop= \left (\sum_{i=0}^{\infty} (\prod_{j = 1}^m\sqrt{(i+j)}) \ket{i}\bra{i+m}\right)\left( \sum_{i=0}^{\infty}\sqrt{i+1} \ket{i}\bra{i+1} \right) \\
                & = \sum_{i=0}^\infty (\prod_{j = 1}^m\sqrt{(i+j)}) (\sqrt{i + m + 1}) \ket{i}\bra{i + m + 1} = \sum_{i=0}^\infty (\prod_{j = 1}^{m+1}\sqrt{(i+j)})\ket{i}\bra{i + m + 1}.
            \end{aligned}
        \end{equation*}
    \end{itemize}
    \item For the creation operator ($\cop$). The proof is similar as above.
    \begin{itemize}
        \item \textit{Base case: $m=0$.} We get again the identity matrix.
        \item \textit{Inductive case.} We show that it also holds for the inductive case $m+1$. Since the Fock basis forms an orthonormal basis we get
        \begin{equation*}
            \begin{aligned}
                (\cop)^{m+1} &=  (\cop)^{m}\cop= \left (\sum_{i=m}^{\infty} (\prod_{j = 0}^{m-1}\sqrt{(i-j)}) \ket{i}\bra{i-m}\right)\left( \sum_{i=0}^{\infty}\sqrt{i+1} \ket{i+1}\bra{i} \right) \\
                & = \sum_{i=m+1}^\infty (\prod_{j = 0}^{m-1}\sqrt{(i-j)}) (\sqrt{i - m}) \ket{i}\bra{i - (m+1)} = \sum_{i=m+1}^\infty (\prod_{j = 0}^{m}\sqrt{(i-j)})\ket{i}\bra{i - (m+1)}.
            \end{aligned}
        \end{equation*}
    \end{itemize}
\end{itemize}
\end{proof}

\subsubsection{Appendix: Characterization of the Schwartz space in the Fock basis}
\label{appendix:prop:domainFock}
\begin{proposition}[Characterization of the Schwartz space in the Fock basis]
Let $\qvec\in \hs$ be a vector with Fock basis expansion $\ket{\qvec} = \sum_{\mnat\in \Nats^\numModes}\cnum_{\mnat}\ket{\mnat}$. Then $\qvec \in \denseDomain$ iff for all $\mvals{m} \in \Nats^\numModes$ we have $\sup_{\mnat \in \Nats^\numModes} \mnat^{\mvals{m}}|\cnum_{\mnat}| < \infty$~\cite{Becnel15}.
\label{prop:domainFock}
\end{proposition}

\subsubsection{Appendix: Unitary equivalence (up to a global phase) }
\label{appendix:prop:unitary_eq}
\begin{proposition}[Unitary equivalence (up to a global phase)]For all unitaries $U_1, U_2 \in \allOps$, there exists a phase $\theta \in [0, 2\pi]$ such that $U_1=e^{\im \theta}\cdot U_2$ iff $U^\dagger_1 \cdot\obsop\cdot U_1 = U^\dagger_2 \cdot\obsop\cdot U_2$ for all canonical observables $\obsop \in \quadOps$.
\label{prop:unitary_eq}
\end{proposition}
\begin{proof} For every canonical observable $\obsop \in \quadOps$ and unitaries $U_1, U_2$ in $\allOps$ we have the following.
\begin{itemize}
    \item ($\Rightarrow$) Suppose there exists a phase $\theta \in [0, 2\pi]$ such that $U_1=e^{\im \theta}\cdot U_2$.
    Then $U^\dagger_1 \cdot \obsop \cdot U_1 = e^{-\im \theta}\cdot e^{\im \theta}\cdot U^\dagger_2 \cdot \obsop \cdot U_2= U^\dagger_2\cdot \obsop \cdot U_2$.
    \item ($\Leftarrow$) Suppose $U^\dagger_1\cdot \obsop'\cdot U_1 = U^\dagger_2\cdot \obsop' \cdot U_2$ for all canonical observables $\obsop' \in \quadOps$.
    With $W = U_2\cdot U_1^\dagger$,  we have
    $$
    W^\dagger \cdot\obsop\cdot W = 
U_1\cdot U_2^\dagger\cdot\obsop\cdot U_2\cdot U_1^\dagger = U_1\cdot (U_2^\dagger\cdot \obsop \cdot U_2)\cdot U_1^\dagger = U_1\cdot U_1^\dagger\cdot \obsop \cdot U_1\cdot U_1^\dagger = \obsop.
    $$
    Thus $W$ commutes with all canonical observables, and hence with the unitary groups generated by $\obsop$.
    By Stone-von Neumann theorem these observables act irreducibly on $\hs$~\cite{Hall13}. So by Schur's Lemma~\cite{Ballentine14}, $W = e^{\im \theta}\cdot \Id$, for some $\theta \in [0, 2\pi]$. Since $W = U_2\cdot U_1^\dagger= e^{\im \theta}\cdot \Id$ it follows $U_2 = e^{\im \theta}\cdot  U_1$.
\end{itemize}
\end{proof}

\subsubsection{Appendix: Invariance and duality under atomic instructions semantics}
\label{appendix:lemma:atomic_schwartz}
\begin{lemma}[Invariance and duality under atomic instructions semantics]
For all atomic instructions $\instruction \in \instructions$, Schwartz cq-states $\hstate \in \hstatesS$, and tempered cq-observables $\hobs \in \allhobsS$, the following hold:
\begin{enumerate}
    \item (Invariance of Schwartz cq-states) $\sem{\instruction}(\hstate) \in \hstatesS$,
    \item (Invariance of tempered cq-observables) $\semD{\instruction}(\hobs) \in \allhobsS$, and 
    \item (Semantics duality) $\Tr\left(\semD{\kset}(\hobs)\otimes \hstate\right) = \Tr\left(\hobs\otimes\sem{\kset}(\hstate)\right) < \infty.$
\end{enumerate}
\label{lemma:atomic_schwartz}
\end{lemma}
\paragraph{Proof (1): Invariance of Schwartz cq-states.}
\begin{proof}
Let $\instruction = \{(\op_i, u_i)\}_{i\in I}$, $\hstate' = \sem{\instruction}(\hstate)$, and $\nop = \nop_{\allqv}$. For a cq-state $\hstate$ and $k \in \Nats$ we write
\[
V_k(\hstate) = \sum_{\cstate\in\cstates}(1+|\cstate|)^k\,\Tr((\nop+\Id)^k\,\hstate(\cstate)),
\]
which corresponds to the Equation~\ref{eq:schwartz_dm} used to define Schwartz cq-states.
By Tonelli's theorem~\cite{evans10} we may exchange sums with traces of positive semidefinite operators:
\[
\sum_{\cstate'}\Tr(\hstate'(\cstate')) = \sum_{\cstate}\sum_{i \in I}\Tr(\op_i^\dagger\op_i\cdot\hstate(\cstate)) = \sum_{\cstate}\Tr(\hstate(\cstate)),
\]
where the last step uses $\sum_{i\in I}\op_i^\dagger\op_i = \Id$ and the normality of the trace. In particular every $\hstate'(\cstate')$ is a positive trace-class operator and $\Tr(\hstate') = \Tr(\hstate) \leq 1$.

Writing $\hstate(\cstate) = \sum_j p_{\cstate,j}\cdot \toproj{\qvec_{\cstate,j}}{}$ in its spectral decomposition with $p_{\cstate,j} > 0$, each eigenvector satisfies $\qvec_{\cstate,j} \in \denseDomain$ (see proof of Lemma~\ref{lemma:inv-dm}).
For a nonnegative integer $k \in \Nats$ we consider the positive and tempered  cq-observable $\hobs_k$, defined as  $\hobs_k(\cstate) = (1+|\cstate|)^k(\nop+\Id)^k$ for all classical states $\cstate \in \cstates$.
Since every Kraus operator leaves $\denseDomain$ invariant\cite{Folland89,Birkhauser06,Gosson11} and all terms are non-negative, Tonelli's theorem and regrouping by $\cstate' = u_i(\cstate)$ give
\[
V_k(\hstate') = \sum_{\cstate \in \cstates}\sum_{i \in I}\sum_j p_{\cstate,j}\cdot\innerProd{\op_i\cdot \qvec_{\cstate,j}}{\hobs_k(u_i(\cstate))\cdot \op_i\cdot \qvec_{\cstate,j}} = \Tr\big(\semD{\instruction}(\hobs_k)\otimes\hstate\big) < \infty,
\]
where finiteness follows from Lemma~\ref{lemma:finite_exp_stemp}, since $\semD{\instruction}(\hobs_k)$ is tempered by~(2), whose proof does not use~(1).
\end{proof}

\paragraph{Proof (2): Invariance of tempered cq-observables.}
\begin{proof}
     Let $\hobs' = \semD{\instruction}(\hobs)$ and fix $\nat \in \Nats$ with constants $C_\nat, k_\nat, b_\nat$ as in Equation~\ref{eq:well_behaved}. In each case, $\hobs'(\cstate) = \hobs'(\cstate)^\dagger$ since every term of the backward semantics has the form $\op_k^\dagger\hobs(u_k(\cstate))\op_k$. 
     We show the bound of Equation~\ref{eq:well_behaved} that defines tempered cq-observables is satisfied by $\hobs'$.
     Consider a Schwartz vector $\qvec \in\denseDomain$.
    \begin{itemize}
        \item Case $\instruction = U \in \unitaries$. The unitaries $U$ and $U^\dagger$ leave $\denseDomain$ invariant~\cite{Folland89,Birkhauser06,Gosson11}. By Equation~\ref{eq:dual_unitaries}, $U^\dagger(\nop+\Id)^\ell U$ and $U(\nop+\Id)^\ell U^\dagger$ are polynomials of degree at most $4\ell$, and a monomial of degree $d$ is bounded on $\denseDomain$ by $(\nop+\Id)^{\lceil d/2\rceil}$ up to a constant. Hence there are constants $D_\ell$ with $\norm{(\nop+\Id)^\ell U\qvec} = \norm{U^\dagger(\nop+\Id)^\ell U\qvec} \leq D_\ell\norm{(\nop+\Id)^{2\ell}\qvec}$, and likewise for $U^\dagger$. Applying Equation~\ref{eq:well_behaved} for $\hobs$ gives
        \[
            \norm{(\nop+\Id)^\nat U^\dagger\hobs(\cstate)U\qvec} \leq D_\nat C_{2\nat} D_{b_{2\nat}}(1+|\cstate|)^{k_{2\nat}}\norm{(\nop+\Id)^{2b_{2\nat}}\qvec}.
        \]
        \item Case $\instruction = \ResetI(\qvar)$. Let $\nop'$ be the number operator of the qumodes in $\allqv\setminus\{\qvar\}$, so that $\nop = \nop_\qvar + \nop'$. Write $\ket{\qvec} = \sum_i \ket{i}\otimes\ket{\qvec_i}$, where $\ket{i} \in \hs_{\qvar}$ and $\ket{\qvec_i} \in \hs_{\allqv \setminus \{\qvar\}}$; and let $\op_0 = \bra{0}_\qvar\hobs(\cstate)\ket{0}_\qvar$. Then $\hobs'(\cstate)\ket{\qvec} = \sum_i\ket{i}\otimes (\op_0\cdot \ket{\qvec_i})$, and $\norm{(\nop'+\Id)^\nat \op_0\qvec_i} \leq \norm{(\nop+\Id)^\nat\hobs(\cstate)(\ket{0}_\qvar\otimes\qvec_i)} \leq C(1+|\cstate|)^{k}\norm{(\nop'+\Id)^{b}\qvec_i}$. Since $\nop + \Id \leq (\nop_\qvar+\Id)(\nop'+\Id)$,
        \[
            \norm{(\nop+\Id)^\nat\hobs'(\cstate)\qvec}^2 \leq \sum_i (i+1)^{2\nat}\norm{(\nop'+\Id)^\nat \op_0\qvec_i}^2 \leq C^2(1+|\cstate|)^{2k}\norm{(\nop+\Id)^{\nat+b}\qvec}^2.
        \]
        \item Case $\instruction = (\MeasI{\cvar}{\qvar})$. Let $B_j = \hobs(\cwrite{\cstate}{\cvar}{j})$. Then $1 + |\cwrite{\cstate}{\cvar}{j}| \leq (1+|\cstate|)(1+j)$, and $(1+j)\toproj{j}{\qvar} = (\nop_\qvar+\Id)\toproj{j}{\qvar} \leq (\nop+\Id)\toproj{j}{\qvar}$. Since the projectors $\toproj{j}{\qvar}$ are orthogonal and commute with $\nop$,
        \begin{equation*}
            \begin{aligned}
                \norm{(\nop+\Id)^\nat\hobs'(\cstate)\qvec}^2 &= \sum_{j=0}^{\infty}\norm{\toproj{j}{\qvar}(\nop+\Id)^\nat B_j\toproj{j}{\qvar}\qvec}^2 \\
                &\leq C^2(1+|\cstate|)^{2k}\sum_{j=0}^{\infty}\norm{(\nop+\Id)^{b+k}\toproj{j}{\qvar}\qvec}^2 \\
                &= C^2(1+|\cstate|)^{2k}\norm{(\nop+\Id)^{b+k}\qvec}^2.
            \end{aligned}
        \end{equation*}
    \end{itemize}
\end{proof}

\paragraph{Proof (3): Semantics duality.}
\begin{proof}
For a cq-observable $\hobs$ and a cq-state $\hstate$ we write $\mathcal{E}(\hobs,\hstate) = \Tr(\hobs\otimes\hstate)$. We show absolute convergence as absolute convergence allows to exchange the order of summation~\cite{Rudin76}.
Let $\instruction = \{(\op_k,u_k)\}_{k\in I}$, $\hstate' = \sem{\instruction}(\hstate)$ and $\hobs' = \semD{\instruction}(\hobs)$. As in the proof of~(1), write $\hstate(\cstate) = \sum_j p_{\cstate,j}\toproj{\qvec_{\cstate,j}}{}$ with $\qvec_{\cstate,j} \in \denseDomain$. Every Kraus operator leaves $\denseDomain$ invariant, so $\op_k\qvec_{\cstate,j} \in \denseDomain$.

Define the triple series with terms $t_{\cstate,k,j} = p_{\cstate,j}\innerProd{\op_k\qvec_{\cstate,j}}{\hobs(u_k(\cstate))\,\op_k\qvec_{\cstate,j}}$. By Equation~\ref{eq:duality_bound}, Tonelli's theorem, and by regrouping the positive terms by $\cstate' = u_k(\cstate)$,
\[
    \sum_{\cstate}\sum_k\sum_j |t_{\cstate,k,j}|
    \leq C_0\sum_\cstate\sum_k(1+|u_k(\cstate)|)^{m}\,\Tr\big((\nop+\Id)^{m}\op_k\hstate(\cstate)\op_k^\dagger\big)
    = C_0\,V_m(\hstate') < \infty,
\]
where $V_m(\hstate') < \infty$ by~(1).
By Fubini's theorem, we may rearrange terms for absolutely convergent series, and follows $\Tr\left(\semD{\kset}(\hobs)\otimes \hstate\right) = \Tr\left(\hobs\otimes\sem{\kset}(\hstate)\right)$.
\end{proof}

\subsubsection{Appendix: Invariance and duality under admissible atomic programs semantics}
\label{appendix:lemma:admissible_atomic}
\begin{lemma}[Invariance and duality under admissible atomic programs semantics]
\label{lemma:admissible_atomic}
For all admissible \myPL~atomic programs $\aprog \in \aprogs$, Schwartz cq-states $\hstate \in \hstatesS$, and tempered cq-observables $\hobs \in \allhobsS$, the following holds:
\begin{enumerate}
    \item (Invariance of Schwartz cq-states) $\sem{\aprog}(\hstate) \in \hstatesS$,
    \item (Invariance of tempered cq-observables) $\semD{\aprog}(\hobs) \in \allhobsS$, and 
    \item (Semantics duality) $\Tr\left(\semD{\aprog}(\hobs)\otimes \hstate\right) = \Tr\left(\hobs\otimes\sem{\aprog}(\hstate)\right) < \infty.$
\end{enumerate}
\end{lemma}
\begin{proof}
If $\aprog \in \instructions$ is an atomic instruction, the claim is Lemma~\ref{lemma:atomic_schwartz}. 
Otherwise $\aprog$ is a unitary with parameters given by real expressions, and it leaves the classical state unchanged. For $\cstate \in \cstates$, let $U_\cstate$ be the unitary obtained by evaluating the parameters of $\aprog$ at $\cstate$; then $\sem{\aprog}(\hstate)(\cstate) = U_\cstate\,\hstate(\cstate)\,U_\cstate^\dagger$ and $\semD{\aprog}(\hobs)(\cstate) = U_\cstate^\dagger\,\hobs(\cstate)\,U_\cstate$. For each fixed $\cstate$, the argument of Lemma~\ref{lemma:atomic_schwartz} applies to $U_\cstate$; it remains to show that the constants it produces grow at most polynomially in $|\cstate|$, since such growth is absorbed by the factors $(1+|\cstate|)^k$ in the definitions of Schwartz cq-states and tempered cq-observables.

By Equation~\ref{eq:dual_unitaries}, $U_\cstate^\dagger\,\dop_\qvar\,U_\cstate$ is a polynomial of degree at most $2$ whose coefficients are bounded in absolute value by $1$ (rotation and beam splitter), by $|\sem{\rexpr_0}(\cstate)| + |\sem{\rexpr_1}(\cstate)|$ (displacement), by $|\sem{\rexpr}(\cstate)|$ (cubic phase), or by $e^{|\sem{\rexpr}(\cstate)|}$ (squeeze). By admissibility, all of these are bounded by $C\cdot(1+|\cstate|)^m$ for some $C, m \in \Nats$. The same holds for $U_\cstate\cdot\dop_\qvar\cdot U_\cstate^\dagger$, since $U_\cstate^\dagger$ is the same gate with negated parameters.
\begin{enumerate}
    \item[(2)] Let $\hobs' = \semD{\aprog}(\hobs)$, so $\hobs'(\cstate) = U_\cstate^\dagger\hobs(\cstate)U_\cstate$ is self-adjoint and preserves $\denseDomain$. Fix $\ell \in \Nats$. By the above, $U_\cstate^\dagger(\nop+\Id)^\ell U_\cstate$ is a polynomial of degree at most $4\ell$ with $\cstate$-independent monomials and coefficients bounded by $C^{2\ell}(1+|\cstate|)^{2m\ell}$; since a monomial of degree $d$ is bounded on $\denseDomain$ by $(\nop+\Id)^{\lceil d/2\rceil}$ up to a constant, there is $D_\ell$ with $\norm{(\nop+\Id)^\ell U_\cstate\qvec} \leq D_\ell(1+|\cstate|)^{2m\ell}\norm{(\nop+\Id)^{2\ell}\qvec}$ for all $\qvec \in \denseDomain$, and likewise for $U_\cstate^\dagger$. Applying Equation~\ref{eq:well_behaved} for $\hobs$ gives
    \[
     \norm{(\nop+\Id)^\ell\hobs'(\cstate)\qvec} \leq D_\ell C_{2\ell} D_{b_{2\ell}}(1+|\cstate|)^{2m\ell+k_{2\ell}+2mb_{2\ell}}\norm{(\nop+\Id)^{2b_{2\ell}}\qvec}.
    \]
    \item[(1)] As in Lemma~\ref{lemma:atomic_schwartz}(1), with $\hobs_k(\cstate) = (1+|\cstate|)^k(\nop+\Id)^k$ we get $V_k(\sem{\aprog}(\hstate)) = \Tr(\semD{\aprog}(\hobs_k)\otimes\hstate) < \infty$ by~(2) and Lemma~\ref{lemma:finite_exp_stemp}.
    \item[(3)] For each classical state $\cstate\in \cstates$, we have $\Tr(U_\cstate^\dagger\hobs(\cstate)U_\cstate\,\hstate(\cstate)) = \Tr(\hobs(\cstate)\,U_\cstate\hstate(\cstate)U_\cstate^\dagger)$ because $U_\cstate$ preserves $\denseDomain$. Both sums over $\cstate$ converge absolutely by~(1), (2) and Lemma~\ref{lemma:finite_exp_stemp}, so they are equal.
\end{enumerate}
\end{proof}

\subsection{Proofs}
\subsubsection{Appendix: Fock-state projectors as polynomial limits (Proposition~\ref{prop:fock_proj_poly})}
\label{appendix:prop:fock_proj_poly}
\begin{proposition*}[Fock-state projectors as polynomial limits]  For all nonnegative integers $\nat \in \Nats$, we have $\toproj{\nat}{}=\lim_{\beta \rightarrow \infty} e^{-\beta\cdot (\nop - \nat\cdot\Id)^2},$ where the limit converges in trace-norm.
\end{proposition*}
\begin{proof}
Let $\poly = (\nop - \nat\cdot\Id)^2 = \nop^2 - 2\nat\cdot\nop + \nat^2\cdot\Id= \sum_{i=0}^{\infty}(i-\nat)^2\toproj{i}{}$.

For all $\beta \geq 0$ and $i \in \Nats$ we have $\poly^k\ket{i} = (i-\nat)^{2k}\ket{i}$ for all $k \in \Nats$, so the partial sums of the exponential series converge on every Fock-state:
\begin{equation*}
    \sum_{k=0}^{\infty}\frac{(-\beta)^k}{k!}\,\poly^k\ket{i} = \left(\sum_{k=0}^{\infty}\frac{(-\beta(i-\nat)^2)^k}{k!}\right)\ket{i} = e^{-\beta(i-\nat)^2}\ket{i}.
\end{equation*}
By linearity, the series converges on the dense subspace of finite linear combinations of Fock-states, where it coincides with the operator
$
    \op_\beta = \sum_{i=0}^{\infty} e^{-\beta(i-\nat)^2}\toproj{i}{}.
$
Since $0 < e^{-\beta(i-\nat)^2} \leq 1$, the operator $\op_\beta$ is bounded.
Since $e^{-\beta(\nat-\nat)^2} = 1$, we get
$
    \op_\beta = \sum_{i \neq \nat} e^{-\beta(i-\nat)^2}\toproj{i}{} + \toproj{\nat}{},
$
which is diagonal in the Fock basis, and hence
$$
    \lim_{\beta\rightarrow\infty} \Tr(|e^{-\beta\cdot\poly} - \toproj{\nat}{}|) =  \lim_{\beta\rightarrow\infty} \sum_{i \neq \nat} e^{-\beta(i-\nat)^2}\toproj{i}{} = 0.
$$
\end{proof}

\subsubsection{Appendix: Invariance of $\denseDomain$ for Schwartz density operators (Lemma~\ref{lemma:inv-dm})}
\label{appendix:lemma:inv-dm}
\begin{lemma*}[Invariance of $\denseDomain$ for Schwartz density operators] 
For all vectors $\qvec \in \hs$, and Schwartz density operators $\dm \in \alldmS$; we have $\dm\ket{\qvec} \in \denseDomain$, and therefore,  $\alldmS \subset \allOpsA$.
\end{lemma*}
\begin{proof}
Let $\dm \in \alldmS$ and denote by  $\nop = \sum_{\modeI = 0}^{\numModes-1} \nop_\modeI$ the total number operator. We need to show that $\dm$ preserves the Schwartz space.
Since $(\nop+\Id)\in\obspolys$, follows $\Tr((\nop+\Id)^k\dm) = \sum_{\mnat \in \mNats{\numModes}} (|\mnat| + 1)^k \cdot \innerProd{\mnat}{\dm\cdot \ket{\mnat}} < \infty$ by the definition of Schwartz density operators, for all nonnegative integers $k \in \Nats$.
For any vector $\qvec \in \hs$ and tuple $\mnat \in \mNats{\numModes}$, Cauchy–Schwarz gives
$\norm{\innerProd{\mnat}{\sqrt{\dm}\cdot \ket{\qvec}}}^2 \le \bra{\mnat}\dm\ket{\mnat} \cdot \norm{\qvec}^2 ,$
so
$$\sum_{\mnat\in \mNats{\numModes}} (|\mnat|+1)^{2k}\cdot\norm{{\innerProd{\mnat}{\sqrt{\dm}\cdot \ket{\qvec}}}}^2 \le  \Tr((\nop+\Id)^{2k}\dm)\cdot \norm{\qvec}^2 < \infty.$$
By Proposition~\ref{prop:domainFock} in Appendix Section~\ref{appendix:prop:domainFock}, it follows that $\dm$ maps the whole Hilbert space into $\denseDomain$.
\end{proof}

\subsubsection{Appendix: Finite expectation values of tempered cq-observables on Schwartz cq-states (Lemma~\ref{lemma:finite_exp_stemp})}
\label{appendix:lemma:finite_exp_stemp}
\begin{lemma*}[Expectation values of tempered cq-observables on Schwartz cq-states] For all Schwartz cq-states $\hstate \in \hstatesS$ and tempered cq-observables $\hobs \in \allhobsS$, holds $\Tr(|\hobs\otimes\hstate|) < \infty$.
\end{lemma*}
\begin{proof}
Write $\nop = \sum_{\modeI = 0}^{\numModes-1} \nop_\modeI$ for the total number operator.
Consider a Schwartz cq-state $\hstate \in \hstatesS$ and a tempered observable $\hobs \in \allhobsS$.
Let $C_0, k, b$ be the constants of Equation~\ref{eq:well_behaved} for $\hobs$ with $\nat = 0$, and let $m = \max(k, 2b)$.
For all $\qvec\in\denseDomain$ and $\cstate'\in\cstates$, since $\norm{\qvec} \leq \norm{(\nop+\Id)^{b}\ket{\qvec}}$ and $(\nop+\Id)^{2b}\leq(\nop+\Id)^{m}$:
\begin{equation}
\norm{\hobs(\cstate')\ket{\qvec}}\cdot\norm{\qvec}\leq C_0(1+|\cstate'|)^{k}\norm{(\nop+\Id)^{b}\ket{\qvec}}^2 \leq C_0(1+|\cstate'|)^{m}\innerProd{\qvec}{(\nop+\Id)^{m}\qvec}.
    \label{eq:duality_bound}
\end{equation}
Write $\hstate(\cstate) = \dm_{\cstate} = \sum_j p_{\cstate,j}\toproj{\qvec_{\cstate,j}}{}$ in its spectral decomposition, where $\dm_{\cstate}$ is a Schwartz density operator and its unit eigenvectors satisfy $\qvec_{\cstate,j} \in \denseDomain$ (Lemma~\ref{lemma:inv-dm}).
For each classical state $\cstate$ and $\obsop = \hobs(\cstate)$, the triangle inequality gives
\begin{equation}
    \begin{aligned}
    \Tr(|{\obsop\dm_{\cstate}}|) &\leq \sum_{j}p_{\cstate,j}\norm{\obsop\ket{\qvec_{\cstate,j}}} \\
    &\leq C_0(1+|\cstate|)^{m}\sum_{j}p_{\cstate,j}\innerProd{\qvec_{\cstate,j}}{(\nop+\Id)^{m}\qvec_{\cstate,j}} 
    \\
    &= C_0(1+|\cstate|)^{m}\Tr({(\nop+\Id)^{m}\dm_{\cstate}}).    
    \end{aligned}
\end{equation}
Follows
\begin{equation}
\begin{aligned}
    \Tr(|\hobs\otimes \hstate|) 
    &= \sum_{\cstate\in\cstates} \Tr(|\hobs(\cstate)\cdot \dm_{\cstate}|)   \leq \sum_{\cstate\in\cstates} C_0(1+|\cstate|)^{m} \Tr((\nop+\Id)^{m}\dm_{\cstate}) < \infty,
\end{aligned}
\end{equation}
where finiteness follows from the definition of Schwartz cq-states.
\end{proof}

\subsubsection{Appendix: Invariance and duality under admissible \myPL~program  semantics (Theorem~\ref{thm:schwartz_invariance})}
\label{appendix:thm:schwartz_invariance}
\begin{theorem*}[Invariance and duality under admissible \myPL~programs semantics]
For all admissible \myPL~programs $\prog \in \progs$, Schwartz cq-states $\hstate \in \hstatesS$, and tempered cq-observables $\hobs\in\allhobsS$; the following holds:
\begin{enumerate}
    \item (Invariance of Schwartz cq-states) $\sem{\prog}(\hstate) \in \hstatesS$,
    \item (Invariance of tempered cq-observables) $\semD{\prog}(\hobs) \in \allhobsS$, and 
    \item (Semantics duality) $\Tr\left(\semD{\prog}(\hobs)\otimes \hstate\right) = \Tr\left(\hobs\otimes\sem{\prog}(\hstate)\right) < \infty.$
\end{enumerate}
\end{theorem*}
\begin{proof}
We prove (1)--(3) simultaneously by structural induction on $\prog$, using three facts:
(a) restricting to a Boolean expression $\bexpr$ preserves Schwartz cq-states and tempered cq-observables, and $\Tr(\hstateAt{\hobs}{\bexpr}\otimes\hstate) = \Tr(\hobs\otimes\hstateAt{\hstate}{\bexpr})$, and 
(b) finite sums of tempered cq-observables are tempered, finite sums of Schwartz cq-states satisfy the bounded of Equation~\ref{eq:schwartz_dm}.
Recall that the equality of tempered cq-observables is a weak equality in $\denseDomain$.
\begin{itemize}
    \item If $\prog \in \aprogs$, the claim is Lemma~\ref{lemma:admissible_atomic}.
    \item If $\prog = \prog_0;\prog_1$, then (1) and (2) follow by applying the induction hypothesis twice, and (3) from
    $\Tr(\semD{\prog_0}(\semD{\prog_1}(\hobs))\otimes\hstate) = \Tr(\semD{\prog_1}(\hobs)\otimes\sem{\prog_0}(\hstate)) = \Tr(\hobs\otimes\sem{\prog_1}(\sem{\prog_0}(\hstate)))$.
    \item If $\prog = \IfElse{\bexpr}{\prog_0}{\prog_1}$, then (1) and (2) follow from the induction hypothesis with (a) and (b), and (3) from
    $\Tr(\semD{\prog}(\hobs)\otimes\hstate) = \sum_{i\in\{0,1\}}\Tr(\semD{\prog_i}(\hobs)\otimes\hstateAt{\hstate}{\bexpr_i}) = \sum_{i\in\{0,1\}}\Tr(\hobs\otimes\sem{\prog_i}(\hstateAt{\hstate}{\bexpr_i})) = \Tr(\hobs\otimes\sem{\prog}(\hstate))$, where $\bexpr_0 = \bexpr$ and $\bexpr_1 = \neg\bexpr$.
    \item If $\prog = \While{\bexpr}{\prog'}$, let $\Psi(\hstate) = \sem{\prog'}(\hstateAt{\hstate}{\bexpr})$ and $\Phi = \Phi_{\bexpr,\prog'}$, and let $k$ be given by admissibility.
    By the induction hypothesis and (a), $\Psi$ preserves $\hstatesS$, $\Phi$ preserves $\allhobsS$, and $\Tr(\Phi^j(\hobs)\otimes\hstate) = \Tr(\hobs\otimes\Psi^j(\hstate))$ for all $j$.
    Hence $\Tr(\hstateAt{\Psi^k(\hstate)}{\bexpr}) = \Tr(\Phi^k(\hstateAt{\hobs_\Id}{\bexpr})\otimes\hstate) = 0$ and $\Psi^{j}(\hstate) = \zeroOp$ for all $j > k$.
    Consequently, $\Tr(\Phi^{k+1}(\hobs)\otimes\hstate) = \Tr(\hobs\otimes\Psi^{k+1}(\hstate)) = 0$ for all $\hstate\in\hstatesS$, so $\Phi^{k+1}(\hobs) = \zeroOp$ for every $\hobs \in \allhobsS$.
    Since 
    $
    \sem{\prog}(\hstate) = \sum_{j=0}^{k}\hstateAt{\Psi^j(\hstate)}{\neg\bexpr}
    $ and $
    \semD{\prog}(\hobs) = \sum_{j=0}^{k}\Phi^j(\hstateAt{\hobs}{\neg\bexpr})
    $, then (1), (2) and (3) follow by induction hypothesis.
\end{itemize}
\end{proof}

\subsubsection{Appendix: Dual map computation via substitutions (Lemma~\ref{lemma:dual_map})}
\label{appendix:lemma:dual_map}
\begin{lemma*}[Dual map computation via substitutions]
    For all admissible atomic programs \(\aprog \in \aprogs\), quantum terms \(\qterm\); if $\sem{\qterm} \in \allhobsS$, then  \(\sem{\atomsubs{\aprog}{\qterm}} = \semD{\aprog}(\sem{\qterm})\).
\end{lemma*}
\begin{proof}
By induction on $\qterm$; we show the equality for every $\qterm$ as weak forms on $\denseDomain$, as for the backward semantics, so temperedness is not needed for the equation itself. We use two facts, both shown by induction on $\qterm$: $\sem{\csubs{\qterm}{\cvar}{\nat}}(\cstate) = \sem{\qterm}(\cwrite{\cstate}{\cvar}{\nat})$, and $\sem{\csubs{\qterm}{\qvar}{\anc}}$ is $\sem{\qterm}$ with the roles of $\hs_\qvar$ and $\hs_\anc$ exchanged. The case $\qterm = \qterm_0+\qterm_1$ follows from linearity of $\semD{\aprog}$.
\begin{itemize}
    \item Case $\aprog$ unitary, with $U_\cstate$ as in the proof of Lemma~\ref{lemma:admissible_atomic}, so $\semD{\aprog}(\hobs)(\cstate) = U_\cstate^\dagger\hobs(\cstate)U_\cstate$. Conjugation by $U_\cstate$ is a $*$-algebra isomorphism, so it acts on polynomials and products by the substitution. It commutes with sums, guards, and $\bra{0}_\anc\cdot\ket{0}_\anc$ (as $U_\cstate$ does not act on $\anc$), and by the functional calculus $U_\cstate^\dagger e^{-\beta\poly}U_\cstate = e^{-\beta\,U_\cstate^\dagger\poly U_\cstate}$, which is preserved by the strong limit $\beta\to\infty$.
    \item Case $\aprog = \MeasI{\cvar}{\qvar}$. Let $P_\nat = \toproj{\nat}{\qvar}$. Then $\semD{\aprog}(\sem{\qterm})(\cstate) = \sum_{\nat=0}^{\infty} P_\nat\,\sem{\qterm}(\cwrite{\cstate}{\cvar}{\nat})\,P_\nat$, and $\sem{\csubs{\kerterm{(\nop_\qvar-\cvar)^2}}{\cvar}{\nat}} = P_\nat$ by Proposition~\ref{prop:fock_proj_poly}. With the first fact, this is the semantics of the two-sided term in the general case. For $\qterm = \rexpr\cdot\monomial_{\mvals{d},\mvals{c}}$: if $\mvals{d}_\qvar \neq \mvals{c}_\qvar$ then $P_\nat\monomial_{\mvals{d},\mvals{c}}P_\nat = \zeroOp$; otherwise $\monomial_{\mvals{d},\mvals{c}}$ is diagonal in the Fock basis of $\qvar$, so $P_\nat\monomial_{\mvals{d},\mvals{c}}P_\nat = P_\nat\monomial_{\mvals{d},\mvals{c}}$, which gives the one-sided term, and $\qterm$ itself if $\cvar$ does not occur in $\rexpr$ since $\sum_{\nat=0}^{\infty} P_\nat = \Id$.
    \item Case $\aprog = \ResetI(\qvar)$. Then $\semD{\aprog}(\hobs)(\cstate) = \Id_\qvar\otimes\bra{0}_\qvar\hobs(\cstate)\ket{0}_\qvar$. For a monomial, $\bra{0}\dop_\qvar^{\mvals{d}_\qvar}(\cop_\qvar)^{\mvals{c}_\qvar}\ket{0} = \delta_{\mvals{d}_\qvar,\mvals{c}_\qvar}\,\mvals{d}_\qvar!$, and $\nat! = \prod_{i=0}^{\nat-1}\sqrt{(i+1)(\nat-i)}$. For $\bindterm{\cvar}{\qterm'}$, $\vacterm{\anc}{\qterm'}$ and $\guardterm{\bexpr}{\qterm'}$, the map $\bra{0}_\qvar\cdot\ket{0}_\qvar$ commutes with sums, guards and, since $\anc\neq\qvar$, with $\bra{0}_\anc\cdot\ket{0}_\anc$; the claim follows by the induction hypothesis. For kernel terms and products, the second fact gives $\bra{0}_\anc\sem{\csubs{\qterm}{\qvar}{\anc}}\ket{0}_\anc = \Id_\qvar\otimes\bra{0}_\qvar\sem{\qterm}\ket{0}_\qvar$.
\end{itemize}
\end{proof}

\subsubsection{Appendix: Soundness and completeness of term judgements (Lemma~\ref{lemma:judgements})}
\label{appendix:lemma:judgements}
\begin{lemma*}[Soundness and completeness of term judgements]
    For all admissible \myPL~programs $\prog \in \progs$ and quantum terms $\qterm, \qterm'$ with $\sem{\qterm} \in \allhobsS $, the following holds: (1)~(soundness) if~$\tjudg{\prog}{\qterm'}{\qterm}$, then $\sem{\qterm'} = \semD{\prog}(\sem{\qterm})$, and (2)~(completeness) if~$\sem{\qterm'} = \semD{\prog}(\sem{\qterm})$, then $\tjudg{\prog}{\qterm'}{\qterm}$.
\end{lemma*}
\begin{proof}
    By induction on the derivation. Every subprogram of an admissible program is admissible, so the induction hypothesis applies to all premises.
    \begin{enumerate}
        \item Soundness.
        \begin{itemize}
        \item \textsc{T-Atom}. By Lemma~\ref{lemma:dual_map}.
        \item \textsc{T-Seq}. By induction hypothesis, $\sem{\qterm''} = \semD{\prog_0}(\sem{\qterm'}) = \semD{\prog_0}(\semD{\prog_1}(\sem{\qterm})) = \semD{\prog_0;\prog_1}(\sem{\qterm})$.
        \item \textsc{T-Eq}. By the induction hypothesis, $\sem{\qterm'} = \sem{\qterm_1'} = \semD{\prog}(\sem{\qterm_1}) = \semD{\prog}(\sem{\qterm})$.
        \item \textsc{T-If}. By the induction hypothesis,
        $\sem{\guardterm{\bexpr}{\qterm_0} + \guardterm{\neg\bexpr}{\qterm_1}}
            = \hstateAt{\semD{\prog_0}(\sem{\qterm_0})}{\bexpr} + \hstateAt{\semD{\prog_1}(\sem{\qterm_1})}{\neg\bexpr}
            = \semD{\IfElse{\bexpr}{\prog_0}{\prog_1}}(\sem{\qterm})$.
        \item \textsc{T-While}. Let $\prog_w = \While{\bexpr}{\prog}$, and denote the initial and successor cq-observables by $\hobs = \sem{\qterm_0}$ and  $\hobs' = \semD{\prog_w}(\hobs)$, respectively.
        For one iteration we get the cq-observable $\hobs_1 = \hstateAt{\semD{\prog}(\hobs)}{\bexpr} + \hstateAt{\hobs}{\neg \bexpr} =\hstateAt{\semD{\prog}(\hobs)}{\bexpr} + \sem{\guardterm{\neg\bexpr}{\qterm_0}}$.
        Let $\qterm$ be a quantum term such that $\sem{\qterm} = \hstateAt{\semD{\prog}(\hobs)}{\bexpr}$, and suppose that the premise is satisfied for $\qterm$, i.e.,   
        $\tjudg{\prog}{\qterm}{\qterm}$.
        So the proof follows since $\hobs' = \semD{\prog_w}(\sem{\qterm}) +  \sem{\guardterm{\neg\bexpr}{\qterm_0}} = \sem{\qterm} + \sem{\guardterm{\neg\bexpr}{\qterm_0}}$ and 
        $\sem{\qterm} = \sem{\guardterm{\qterm}{\bexpr}}$.
    \end{itemize}
    \item Completeness.
    \begin{itemize}
        \item $\prog = \aprog \in \aprogs$, we use Rule \textsc{T-Atom}.
        \item $\prog \in \{\prog_0;\prog_1,~\IfElse{\bexpr}{\prog_0}{\prog_1}\}$ follows by induction hypothesis.
        \item $\prog = \While{\bexpr}{\prog'}$. 
        let $\qterm_0 = \guardterm{\neg\bexpr}{\qterm}$.
        By induction hypothesis, for each iteration $i > 0$, we can derive $\tjudg{\prog'}{\qterm_i''}{\qterm_i}$, where $\qterm_{i+1} = \guardterm{\bexpr}{\qterm_i''} + \guardterm{\neg\bexpr}{\qterm}$.
        Consider a loop that halts at iteration $k > 0$.
        We can use rule \textsc{T-Eq} to derive $\tjudg{\prog'}{\qterm_k''}{\guardterm{\bexpr}{\qterm_k''} + \guardterm{\neg\bexpr}{\qterm}}$ from $\tjudg{\prog'}{\qterm_k''}{\qterm_k}$, and then rule \textsc{T-While} derives $\tjudg{\prog}{\guardterm{\bexpr}{\qterm_k''} + \guardterm{\neg\bexpr}{\qterm}}{\qterm}$.
    \end{itemize}
    \end{enumerate}
\end{proof}

\subsubsection{Appendix: Proof system (Theorem~\ref{thm:proof_system})}
\label{appendix:thm:proof_system}
\begin{theorem*}[Proof system] Let $\vdash$ denote syntactic derivability using the proof system of Figure~\ref{fig:proof_system}. For all admissible \myPL~quantum programs $\prog \in \progs$ and well-formed assertions $\pre, \post \in \assertions$, the following hold: (1)~if $\vdash \hltriple{\pre}{\prog}{\post}$, then $\hltriple{\pre}{\prog}{\post}$ is valid; 
    (2)~the rule \textsc{Prog} derives the weakest liberal precondition of $\prog$ with respect to $\post$; and (3)~if $\hltriple{\pre}{\prog}{\post}$ is valid, then $\vdash\hltriple{\pre}{\prog}{\post}$ relative to oracles deciding the entailment $\Rightarrow$ of assertions.
\end{theorem*}
\begin{proof}
    Let $\prog \in \progs$ be an admissible program and $\post \in \assertions$ a well-formed assertion.
    Given judgements $\tjudg{\prog}{\qterm^*}{\qterm}$ for all $\qterm \in \Terms{\post}$, let $\pre$ be the assertion obtained from $\post$ by replacing every $\qterm \in \Terms{\post}$ by $\qterm^*$, as in rule \textsc{Prog}.
    Then $\pre$ is a well-formed assertion by Theorem~\ref{thm:schwartz_invariance}.
    Also, by Lemma~\ref{lemma:judgements} is easy to show in the syntax of assertions that for all Schwartz cq-states $\hstate \in \hstatesS$, 
    \begin{equation}
        \hstate \vDash \pre \quad\text{iff}\quad \sem{\prog}(\hstate) \vDash \post.
        \label{eq:prog_iff}
    \end{equation}    
    \begin{enumerate}
        \item \textbf{Soundness.} Rule \textsc{Prog} is sound by Equation~\ref{eq:prog_iff}.
         Let $\hstate \in \hstatesS$ with $\hstate \vDash \pre$, and suppose that $\pre \Rightarrow \pre'$, $\hstate \vDash \pre'$. By the induction hypothesis, $\hltriple{\pre'}{\prog}{\post'}$ is valid, so $\sem{\prog}(\hstate) \vDash \post'$, and since $\post' \Rightarrow \post$, $\sem{\prog}(\hstate) \vDash \post$, Rule \textsc{Conseq} is also sound.
        \item \textbf{Weakest liberal precondition derivability.} By Lemma~\ref{lemma:judgements}(2) which shows completeness for quantum terms and Equation~\ref{eq:prog_iff}, Rule \textsc{Prog} derives the weakest liberal precondition.
        \item \textbf{Relative completeness.} Let $\hltriple{\pre}{\prog}{\post}$ be valid, and suppose that the weakest liberal precondition $\pre'$ is derivable. Then $\pre \Rightarrow \pre'$, and rule \textsc{Conseq} derives $\vdash \hltriple{\pre}{\prog}{\post}$.
    \end{enumerate}
\end{proof}

\subsection{Appendix: Quantum state teleportation}
\label{appendix:state_tel}
Every density operator $\dm$ has a unique {\em characteristic function} 
$
\charactF(\mvals{x}, \mvals{p}) = \Tr(\weylOp(\mvals{x}, \mvals{p})\cdot\dm)
$,
where $\mvals{x}, \mvals{p} \in \mReals{\numModes}$ are tuples of real numbers, and $\weylOp(\mvals{x}, \mvals{p})=  \exp(\im\cdot(\sum_{\modeI= 0}^{\numModes-1}\mvals{x}_\modeI\cdot \X_\modeI + \mvals{p}_{\modeI}\cdot\POp_{\modeI}))$ denotes the {\em Weyl operator}~\cite{Weedbrook12}. The characteristic function is well-defined because the Weyl operator is bounded and $\dm$ is trace-class.

The characteristic function $\charactF_{in}(x, p)$ for an arbitrary input state $\dm$ on $\qvar_0$ may be written as $\charactF_{in}(x, p) = \Tr( \exp(\im(x \cdot \X_{\qvar_0} + p\cdot \POp_{\qvar_0}))\cdot\dm)$. 
 Whereas, the characteristic function $\charactF_b(x, p)$ of Bob's state at the end of the program is  
{\small \[
 \begin{aligned}
     \charactF_b(x, p) &= \Tr\left( 
        \exp\left(
            \im\cdot\left(x \cdot \X_{\qvar_2} + p\cdot \POp_{\qvar_2}\right)
        \right)\cdot 
        \sem{\prog}\left(\dm\right)
        \right) 
      = \Tr\left( 
        \exp\left(
            \im\cdot\left(
                x \cdot \semD{\prog}\left(\X_{\qvar_2}\right) + p\cdot  \semD{\prog}\left(\POp_{\qvar_2}\right)
            \right)
        \right)
        \dm
        \right) \\
     &= \Tr\left( 
            \exp\left(
                \im\cdot \left(x\cdot 
                    \left(\X_{\qvar_0} + \sqrt{2}\cdot e^{-r}\cdot \X_{\qvar_2}\right) + 
                    p \cdot \left(\POp_{\qvar_0} + \sqrt{2}\cdot e^{-r}\cdot\POp_{\qvar_1}\right)
                \right)
            \right)\cdot \dm
        \right).
 \end{aligned}
\]}
Since the characteristic function uniquely determines a continuous-variable quantum state, Bob's final state coincides with Alice's input state in the infinite-squeezing limit.
An alternative proof can be obtained via the Choi-Jamiołkowski isomorphism~\cite{Choi75}, by exploiting the fact that the shared resource is a Gaussian state completely characterized by its first and second moments~\cite{Weedbrook12}.